\documentclass[11pt,reqno]{amsart}

\usepackage{amsmath,amsthm,amssymb,physics}
\usepackage[colorlinks=true, linkcolor=black, urlcolor=black, citecolor=black, hypertexnames=false]{hyperref}
\usepackage{xcolor}
\usepackage{caption}
\usepackage{subcaption}
\usepackage{parskip}
\usepackage[absolute,overlay]{textpos}

\definecolor{darkblue}{HTML}{1F4FD8}

\usepackage{amssymb}
\usepackage{amsmath}
\usepackage{enumerate}
\usepackage{booktabs}
\usepackage{dsfont}
\usepackage{pgfplots}
\usepackage{nicefrac}
\usepackage{rotating}
\usepackage{siunitx}

\usepackage[round,authoryear]{natbib}

\newtheorem{theorem}{Theorem}[section]
\newtheorem{lemma}[theorem]{Lemma}
\newtheorem{proposition}[theorem]{Proposition}

\theoremstyle{definition}

\newtheorem{assumption}[theorem]{Assumption}
\newtheorem{example}[theorem]{Example}

\theoremstyle{remark}
\newtheorem*{remark}{Remark}       

\renewcommand{\P}{\mathbb{P}}
\newcommand{\R}{\mathbb{R}}

\newcommand{\N}{\mathbb{N}}
\newcommand{\E}{\mathbb{E}}

\newcommand{\F}{\mathcal{F}}

\newcommand{\A}{\mathcal{A}}

\newcommand{\1}{\mathbf{1}}

\renewcommand{\l}{\left}
\renewcommand{\r}{\right}

\allowdisplaybreaks

\title[Optimal Market Making in Prediction Markets]{Optimal Market Making in Prediction Markets}

\author{Dominik Feil}
\address{University of Konstanz, Department of Mathematics and Statistics, 78457 Konstanz, Germany}
\email{dominik.feil@uni-konstanz.de}

\author{Max Nendel}
\address{University of Waterloo, Department of Statistics and Actuarial Science, Waterloo, ON N2L 3G1, Canada}
\email{mnendel@uwaterloo.ca}

\date{\today}
\thanks{The authors thank Michael Kupper for valuable comments as well as his guidance and support related to this work.\ The second-named author gratefully acknowledges financial support from the Natural Sciences and Engineering Research Council of Canada via Discovery Grant no.\ RGPIN-2025-04219}

\begin{document}
\begin{abstract}
   Prediction markets are attracting growing attention as trading volumes rise and their practical relevance increases.\ To ensure efficient price discovery, liquidity provision becomes ever more important.\ Due to the binary settlement structure in prediction markets, optimal market making leads to an optimization problem that is fundamentally different from the ones studied in classical settings.\ In this paper, we develop a stochastic control framework for prediction markets in which the market price is modeled as a conditional probability of the outcome that is generated by a transformed latent belief diffusion.\ A market maker selects bid and ask quotes to maximize expected terminal wealth while controlling both mark-to-market inventory risk and the settlement risk of remaining positions at resolution.\ We derive the associated Hamilton--Jacobi--Bellman equation and characterize the unique optimal bid and ask quotes.\ By transforming the equation to the latent belief space and using a fixed-point argument, we prove existence and uniqueness of a classical solution and verify the resulting optimal quoting strategy.\ In addition, we provide a numerical analysis, which reveals how optimal liquidity provision in prediction markets depends on inventory, market beliefs, time to resolution, and risk aversion. Further, we demonstrate that the optimal quoting strategy substantially improves downside protection while preserving most of its expected profit relative to a myopic benchmark that maximizes the instantaneous expected mark-to-market profit.

    \smallskip

\end{abstract}

\maketitle
\section{Introduction}
In this article, we study optimal market making in prediction markets. Prediction markets are financial exchanges where participants trade contingent claims on uncertain future events, and a standard contract pays \$1 if the event occurs at or prior to a specified resolution time and \$0 otherwise.\ To facilitate trading and support efficient price discovery, prediction markets require a sufficient supply of liquidity.\ Understanding how a market maker should quote in such markets is therefore both economically and mathematically relevant.\ The aim of this paper is to develop a market making model tailored to the distinctive features of prediction markets and to derive the corresponding optimal quoting strategy.

\subsection{Prediction Markets}
Due to the binary settlement structure of prediction markets, a contract's trading price is commonly regarded as the market-implied probability of the underlying event.\footnote{To reduce opportunity costs and thereby enable accurate long-term forecasting, exchanges may remunerate locked collateral.\ For example, as of July 19, 2026, \emph{Polymarket} offers an annualized rate of 3.25\% on eligible positions in certain long-dated markets, such as ``Presidential Election Winner 2028.''}\ For instance, consider a contract paying \$1 if candidate A wins an upcoming election and \$0 otherwise.\ If this contract trades at 30\textcent, its price is commonly interpreted as implying a 30\% probability of candidate A's victory.

To allow participants to express views in both directions, prediction markets list complementary contracts. In addition to the claim described above, there is a second security paying \$1 if candidate A does not win and \$0 otherwise. These two contracts correspond to mutually exclusive and exhaustive outcomes, and their payoffs sum to one in every state.

Prediction markets have strong predictive power across a variety of forecasting tasks.\ For instance, \citet{BergNelsonRietz2008} show that prediction markets outperform opinion polls in forecasting the vote shares of the two major parties in U.S.\ presidential elections.\ More recently, similar evidence was reported for the 2024 U.S.\ presidential election, where prediction market returns were found to predict subsequent polling data \citep{NgEtAl2026}. Beyond electoral forecasting, prediction markets have also been used successfully to predict infectious disease activity \citep{PolgreenEtAl2007} and to support corporate decision-making processes \citep{CowgillZitzewitz2015}.

    Today, prediction market prices provide probability estimates for a broad range of future events, for example, related to elections, weather and climate, geopolitics, monetary policy, sports, technological developments, and corporate earnings reports.

\subsubsection{Automated Market Makers}
An important strand of the prediction market literature studies automated market makers (AMMs), cf.\ \citet{Hanson2003, ChenPennock2007, AbernethyEtAl2013}.\ Rather than relying on direct matching between traders, such mechanisms provide liquidity through a predetermined pricing rule.

A standard formulation is based on a cost function
\[
    C\colon\R^n \to \R,
\]
defined on the vector of outstanding shares $q=(q_1,\dots,q_n)$, where $q_i$ denotes the total number of shares of outcome $i$ held by traders. Marginal prices are then given by the gradient
\[
    p(q) := \nabla C(q),
\]
and the cost of purchasing a bundle $\Delta$ is
\[
    C(q+\Delta)-C(q).
\]
For a binary event, we have $n=2$, corresponding to one outcome and its complement.

\subsubsection{The Shift to Limit Order Books}
While cost-function-based automated market makers played a central role in early prediction market design, major platforms now rely on limit order books. In a limit order book, participants post buy and sell orders specifying prices and quantities. Transactions occur when incoming orders match standing quotes.

This market design differs fundamentally from cost-function-based automated market makers. In AMM-based systems, prices are determined by a pricing rule and adjust as traders transact with the mechanism. When new information arrives, agents must trade in order to move prices towards the new fair value.\ By contrast, in a limit order book, liquidity providers can revise or cancel their quotes in response to new information.\ Prices may therefore change even in the absence of trades. The bid-ask spread, which compensates liquidity providers for bearing risk, varies over time as market conditions evolve. Liquidity provision thus becomes an active, profit-seeking activity rather than a deterministic component of the market design.

Prediction market platforms operating through limit order books, such as \emph{Kalshi} and \emph{Polymarket}, have grown substantially in recent years. For example, cumulative trading volume in \emph{Kalshi}'s 2024 U.S.\ presidential election winner market exceeded \$500 million, while \emph{Polymarket}'s corresponding market recorded about \$3.7 billion in traded volume.

The efficiency and stability of such markets depend on market makers who supply liquidity.

\subsection{Market Making in Limit Order Books}
A market maker provides liquidity by continuously posting bid and ask quotes at which they are willing to buy and sell a given asset. In doing so, the market maker may earn the bid-ask spread, but is also exposed to inventory risk.\ To manage this risk and adapt to new information, the market maker dynamically adjusts quotes in response to inventory and market conditions.

\subsubsection{The Market Making Problem}
The central problem of a market maker is to choose bid and ask quotes so as to maximize expected profit while controlling risk. A fundamental trade-off arises between the spread captured per transaction and the execution frequency. Tighter quotes increase the execution intensity but reduce the profit per trade, whereas wider quotes increase the margin but reduce the likelihood of execution.

In addition, inventory considerations create a source of quote asymmetry. A market maker with a long inventory position typically quotes more aggressively on the ask side and less aggressively on the bid side, while the opposite adjustment applies to a short position. This asymmetry is referred to as skew.

\subsubsection{Optimal Market Making Models}
The decision problem of a market maker naturally leads to a stochastic control formulation in which bid and ask quotes serve as controls, while the price, inventory, cash position, and possibly additional variables constitute the state of the system.

Building on \citet{HoStoll1981}, the seminal paper by \citet{AvellanedaStoikov2008} studies the optimal control problem of a single-asset market maker who seeks to maximize the expected utility of terminal wealth under constant absolute risk aversion. The authors model the mid-price process as a Brownian motion with constant volatility and assume that liquidity-taking buy and sell orders arrive with intensities of the form $A\exp(-k\delta)$, where $A$ and $k$ are positive constants and $\delta$ denotes the quote offset relative to the mid-price.

\citet{GueantLehalleFernandezTapia2013} modify this framework by introducing inventory limits and show that the resulting four-dimensional Hamilton--Jacobi--Bellman equation can be reduced to a system of linear ordinary differential equations. They also derive closed-form approximations of the optimal quotes. Moreover, \citet{Gueant2017} develops an extension to multi-asset market making.

Further contributions expand the literature on optimal market making along several dimensions.\ \citet{GuilbaudPham2013} allow the market maker to trade using both limit and market orders, while \citet{CarteaJaimungal2015} propose risk metrics for assessing and fine-tuning high-frequency trading strategies.\ To capture richer order flow dynamics, \citet{CarteaJaimungalRicci2014} introduce a mutually exciting process to allow for feedback effects in market orders, while \citet{Jusselin2021} studies optimal market making under order flow driven by general Hawkes processes.\ Other extensions address model uncertainty \citep{NystromOuldAlyZhang2014,CarteaDonnellyJaimungal2017} and options market making \citep{AbergelElAoud2015,BaldacciBergaultGueant2021}. More recently, \citet{BarzykinBergaultGueant2023} allow the market maker to hedge inventory in an external liquidity pool, subject to execution costs and market impact.

To the best of our knowledge, optimal market making in prediction markets has not yet been studied within a stochastic control framework.

\subsection{Main Contributions}
We develop a market making model tailored to prediction markets. In contrast to the classical literature, we treat prices as conditional probabilities taking values in $(0,1)$. The price dynamics are generated through a nonlinear transformation of a latent belief process that aggregates information over time and is specified so that the resulting price process is a martingale.\ The volatility of the price process is allowed to depend on both time and the current price.\ At a fixed terminal time, the event resolves according to a Bernoulli random variable whose distribution is determined by the market's terminal belief. Order arrival intensities depend on time, the current price, and the quotes, which are constrained to lie in the interval $[0,1]$.

Within this framework, we formulate the market making problem as a stochastic optimal control problem in which the agent seeks to maximize expected terminal wealth subject to a running inventory penalty and a terminal settlement risk penalty.\ We formulate the associated Hamilton--Jacobi--Bellman equation and reduce its dimensionality from four to three.\ Theorem \ref{thm.ex.unique.classical} establishes the existence of a classical solution by transforming the equation into the latent space.\ In Theorem \ref{thm:existence_uniqueness_mild}, we first prove existence and uniqueness of a sufficiently regular mild solution using a fixed-point argument and a priori estimates for the H\"older norm, which then allows us to invoke the results of \cite[Chapter 9]{Krylov1996} to show that the mild solution is a classical solution.\ Uniqueness of the classical solution is deduced from uniqueness of the mild solution by a standard argument, showing via It\^o's formula that every classical solution is a mild solution.\ The existence and uniqueness of a classical solution to the Hamilton--Jacobi--Bellman equation enables us to derive and verify the optimal quoting strategy, cf.\ Proposition \ref{lem:H_unique_maximizer_pred_model_general} and Theorem \ref{thm:verification}, respectively.

We then study the model quantitatively by numerically approximating the solution and investigating the behavior of the optimal quotes.\ In particular, we examine how the optimal strategy is shaped by inventory, time to settlement, price, and risk aversion.\ The analysis highlights structural features of prediction markets, including skew effects arising from asymmetric intensity specifications and the diminishing importance of inventory risk as prices approach zero or one.

Finally, we run a Monte Carlo simulation to compare the optimal strategy with a myopic baseline strategy that always quotes to maximize the instantaneous expected mark-to-market profit. We find that the optimal strategy achieves a substantial reduction in risk while sacrificing only a small portion of expected profit.

The remainder of this paper is organized as follows. Section~\ref{sec:market_model_and_control_problem} introduces the prediction market model, the optimization problem, and the corresponding Hamilton--Jacobi--Bellman equation.\ Section~\ref{sec:main_results} derives and verifies the market maker's optimal quoting strategy, cf.\ Proposition \ref{lem:H_unique_maximizer_pred_model_general} and Theorem \ref{thm:verification}, respectively.\ Moreover, we provide sufficient conditions to establish existence and uniqueness of a classical solution to the Hamilton--Jacobi--Bellman equation, cf.\ Theorem \ref{thm.ex.unique.classical}.\ Section~\ref{sec:numerical_analysis_optimal_quotes} presents a numerical analysis of the resulting strategy.
Finally, Section~\ref{sec:conclusion} concludes.

\section{Market Model and Control Problem}\label{sec:market_model_and_control_problem}
In this section, we develop a stochastic control approach to optimal market making in prediction markets that operate via a limit order book. We consider a market maker who provides liquidity for two complementary contracts written on an event.\ One contract pays \$1 if the event occurs and \$0 otherwise, while the other pays \$1 if the event does not occur and \$0 otherwise.\ On the exchange, a bid for one contract is also recorded as an ask for the complementary contract, and vice versa. The problem therefore reduces to liquidity provision in a single contract with short selling allowed.

\subsection{The Model}\label{sec:model}
Throughout, let $T>0$ and $(\Omega,\F,\mathbb{F},\P)$ be a filtered probability space carrying a standard Brownian motion $(W_t)_{t\in[0,T]}$, {adapted to the filtration $\mathbb{F}=(\F_t)_{t\in[0,T]}$, which is assumed to satisfy the usual conditions.} 

We consider a security written on an event that resolves immediately after time $T$. The contract is traded on the time interval $[0,T]$ and its settlement is modeled as a random variable $Y$ taking values in $\{0,1\}$.\ 

We assume that there exists an $\mathcal{F}_T$-measurable random variable $p_T \in (0,1)$ such that
\begin{equation*}
    \P(Y=1 \mid \F_T)=p_T
    \qquad\text{and}\qquad
    \P(Y=0 \mid \F_T)=1-p_T.\footnote{Alternatively, one may enlarge the probability space to support a random variable $U$, uniformly distributed on $(0,1)$ and independent of $\F_T$, and define $Y:=\1_{\{U\le p_T\}}$.}
\end{equation*}
Thus, $p_T$ represents the probability of the tradable outcome immediately before the outcome is revealed. The information available after resolution is described by
\[
    \mathcal G_T:=\mathcal F_T\vee\sigma(Y).
\]
In particular, $Y$ is $\mathcal G_T$-measurable but not $\mathcal F_T$-measurable.

We denote the probability process of the outcome by
\begin{equation*}
    p_t := \E[p_T \mid \F_t] \qquad \text{for } t\in[0,T].
\end{equation*}
Then $(p_t)_{t\in[0,T]}$ is an $\mathbb{F}$-martingale with values in $(0,1)$. By construction, $p_t$ represents the probability of a payout of one, given the information available at time $t$. Assuming market efficiency, we identify $p_t$ with the market price of the contract. We will also refer to $p_t$ as the market belief at time $t\in[0,T]$.

\subsubsection{Information and Price Dynamics}
We first model a real-valued latent belief process $(L_t)_{t\in[0,T]}$, which is then transformed into the market belief process $(p_t)_{t\in[0,T]}$ via a nonlinear map $f\colon\R\to(0,1)$ such that
\[
    p_t := f(L_t)
\]
for all $t\in[0,T]$. We assume that $f\in C^2\big(\R;(0,1)\big)$ satisfies
\[
    \lim_{x\to-\infty} f(x)=0, 
    \qquad 
    \lim_{x\to\infty} f(x)=1,
\]
and $f'(x)>0$ for all $x\in\R$. In addition, defining $g\colon\R\to\R$ by
\[
    g(x) := \frac{f''(x)}{f'(x)},
\]
we assume that $g$ is bounded and Lipschitz continuous.\ Notice that $g$ coincides, up to sign, with the usual Arrow--Pratt coefficient of absolute risk aversion. Under these assumptions, both $f'$ and $f''$ are bounded and Lipschitz continuous, as shown in Lemma \ref{lem:f'_bounded}.

The latent belief $L_t$ can be interpreted as a real-valued summary of the market's aggregate information at time $t$ about the settlement parameter, while the map $f$ converts the real-valued information summary into a price. We model $(L_t)_{t\in[0,T]}$ via the stochastic differential equation
\begin{equation}\label{eq:sde.latent}
    dL_t = \mu(t,L_t)\,dt + \sigma(t,L_t)\,dW_t,
\end{equation}
where $\mu\colon [0,T]\times\R\to\R$ is a priori measurable and $\sigma\colon [0,T]\times\R\to(0,\infty)$ is continuous.\ Moreover, we assume that there exists a constant $C>0$ such that
\begin{equation}\label{eq:lipschitz_coefficients}
    |\sigma(t,x)-\sigma(t,y)| \le C|x-y|\quad \text{and}\quad |\sigma(t,x)| \le C
\end{equation}
for all $t\in [0,T]$ and $x,y\in \R$.

In a first step, we derive the appropriate choice of $\mu$ to ensure that $(p_t)_{t\in[0,T]}$ satisfies the martingale property. Formally applying Itô's formula to $p_t=f(L_t)$ yields
\[
dp_t
=
\big(
f'(L_t)\mu(t,L_t)
+\tfrac12 f''(L_t)\sigma(t,L_t)^2
\big)\,dt
+
f'(L_t)\sigma(t,L_t)\,dW_t.
\]
For $(p_t)_{t\in[0,T]}$ to be a martingale, its drift must vanish. Hence, we set
\begin{equation}\label{eq:def.mu}
\mu(t,x)
:=
-a(t,x)g(x)\quad \text{with}\quad a(t,x):=\frac12 \sigma(t,x)^2
\end{equation}
for all $(t,x)\in[0,T]\times\R$.\ Since $\sigma$ is bounded and Lipschitz continuous in the space variable, it follows that $a$ is bounded and Lipschitz continuous in the space variable.\ Therefore, $\mu$ is bounded and Lipschitz continuous in the space variable as $g=\frac{f''}{f'}$ is, by assumption, bounded and Lipschitz continuous.\ We then obtain
\[
dp_t = f'(L_t)\sigma(t,L_t)\,dW_t = \varsigma(t,p_t)\,dW_t,
\]
where $\varsigma\colon[0,T]\times(0,1)\to\R$ is given by
\begin{equation}\label{eq:def.varsigma}
    \varsigma(t,p) := f'\!\l(f^{-1}(p)\r)\sigma\!\l(t,f^{-1}(p)\r)\!.
\end{equation}
Since $f'$ and $\sigma$ are both bounded, $\varsigma$ is bounded as well.\ In addition, since $g=\frac{f''}{f'}$ is bounded and $\sigma$ is bounded and globally Lipschitz continuous in the space variable, one readily verifies that $\varsigma$ is globally Lipschitz continuous in the market belief $p\in (0,1)$.\ Moreover, $(p_t)_{t\in[0,T]}$ is a bounded continuous local martingale, and hence a martingale.

Due to assumption \eqref{eq:lipschitz_coefficients} and the definition of $\mu$ in \eqref{eq:def.mu}, for every $t\in [0,T]$ and $x\in \R$, the stochastic differential equation \eqref{eq:sde.latent} has a unique strong solution $(L_s^{t,x})_{s\in [t,T]}$ with $L_t^{t,x}=x$ $\P$-a.s.\ Since the function $f$ is bijective, the same holds for the dynamics of the market price, and we use the notation $(p_s^{t,p})_{s\in [t,T]}$ for the market price starting from $p\in (0,1)$ at time $t\in [0,T]$, i.e., $p_t^{t,p}=p$ $\P$-a.s. 

A natural choice for the transformation $f$ is the logistic function, which satisfies all of the previously stated assumptions.
\begin{example}[Logistic Transformation]
Consider the logistic function $f\colon\R\to (0,1)$ given by
\[
f(x):=\frac{1}{1+e^{-x}}.
\]
Indeed, we have $f\in C^2\big(\R;(0,1)\big)$,
\[
\lim_{x\to-\infty} f(x)=0, \qquad \lim_{x\to+\infty} f(x)=1,
\]
and
\[
    f'(x)=f(x)\big(1-f(x)\big)>0
\]
for all $x\in\R$. Moreover,
\[
    f''(x)=f'(x)\big(1-2f(x)\big),
\]
so that
\[
g(x):=\frac{f''(x)}{f'(x)}=1-2f(x)=-\tanh\Big(\frac{x}{2}\Big).
\]
In particular, $g$ is bounded and Lipschitz continuous.\ Hence, the logistic function satisfies all assumptions imposed on $f$.

In this case, the latent belief process is given by
\[
    dL_t=\frac12 \tanh\l(\frac{L_t}{2}\r)\sigma(t,L_t)^2\,dt
    +\sigma(t,L_t)\,dW_t,
\]
and the corresponding price process evolves according to
\[
    dp_t=p_t(1-p_t)\sigma\!\l(t,\ln\frac{p_t}{1-p_t}\r)\,dW_t.
\]
Moreover, the latent belief at time $t$ is given by
\[
    L_t=\ln\frac{p_t}{1-p_t},
\]
and therefore coincides with the log-odds of the market-implied probability $p_t$.\ This is closely related to the logarithmic market scoring rule introduced by \citet{Hanson2003}, under which, in the binary case, the log-odds are proportional to the net outstanding shares.\ The recent work by \citet{Dalen2026}
likewise adopts the price log-odds as a real-valued belief.
\end{example}

\subsubsection{Cash and Inventory Dynamics}
At any time $t$, the market maker posts bid and ask quotes $\pi_t^b, \pi_t^a\in[0,1]$ at which the agent is willing to buy or sell the traded security. We assume that trades occur in fixed sizes $\Delta>0$. Let $(N^b_t)_{t\in[0,T]}$ and $(N^a_t)_{t\in[0,T]}$ be counting processes representing executions at the bid and ask, respectively. Their construction is specified in Section~\ref{sec:order_intensities}.

The cash process $(X_t)_{t\in[0,T]}$ evolves according to
\[
    dX_t = \Delta\pi_t^a\, dN^a_t
    - \Delta\pi_t^b\,dN^b_t,
\]
while the inventory process $(q_t)_{t\in[0,T]}$ satisfies
\[
    dq_t = \Delta\, dN_t^b - \Delta\, dN_t^a.
\]
We impose an inventory constraint by restricting $q_t$ to the finite grid
\[
    \mathcal{Q}:=\big\{-Q,-Q+\Delta,\ldots,Q-\Delta,Q\big\},
\]
where $Q>0$ is an arbitrary multiple of $\Delta$.

\subsubsection{Order Intensities}\label{sec:order_intensities}
We assume that the filtered probability space additionally supports
two mutually independent $\mathbb F$-Poisson random measures $M^b$ and $M^a$ on
$[0,T]\times[0,\infty)$, each
with compensator $ds\, dz$, and independent of the Brownian motion $W$.

For an $\mathbb{F}$-predictable quoting strategy
$\pi=(\pi^b,\pi^a)$ taking values in $[0,1]^2$, the bid and ask
execution processes are defined by
\[
\begin{aligned}
N_t^b
&=
\int_{(0,t]}\int_{[0,\infty)}\1_{\{z\le \Lambda^b(s,p_s,\pi_s^b)\}}\1_{\{q_{s-}<Q\}}\,M^b(ds,dz),\\
N_t^a
&=
\int_{(0,t]}\int_{[0,\infty)}\1_{\{z\le \Lambda^a(s,p_s,\pi_s^a)\}}\1_{\{q_{s-}>-Q\}}
\,M^a(ds,dz),
\end{aligned}
\]
where
\[
q_t=q_0+\Delta N_t^b-\Delta N_t^a.
\]
Thus, the execution processes have the $\mathbb{F}$-predictable
intensities
\[
\lambda_t^b
=
\Lambda^b(t,p_t,\pi_t^b)\mathbf 1_{\{q_{t-}<Q\}},
\qquad
\lambda_t^a
=
\Lambda^a(t,p_t,\pi_t^a)\mathbf 1_{\{q_{t-}>-Q\}},
\]
respectively, for two functions
\[
\Lambda^b,\Lambda^a
\colon
[0,T]\times(0,1)\times[0,1]\longrightarrow[0,\infty).
\]
The functions $\Lambda^b$ and $\Lambda^a$ are assumed to satisfy the following conditions:

\begin{enumerate}[(i)]
\item $\Lambda^b$ and $\Lambda^a$ are continuous and uniformly bounded,
\item $\Lambda^b$ and $\Lambda^a$ are twice continuously differentiable in $\pi$ on $(0,1)$,
\item for every $(t,p,\pi)\in[0,T]\times(0,1)\times(0,1)$, we have
\[
    \partial_\pi \Lambda^b(t,p,\pi)>0
    \qquad\text{and}\qquad
    \partial_\pi \Lambda^a(t,p,\pi)<0,
\]
\item $\Lambda^b$ and $\Lambda^a$ satisfy the curvature condition
\[
    \sup_{\pi\in(0,1)}
    \frac{\Lambda^\circ(t,p,\pi)\,\partial_{\pi\pi}^2 \Lambda^\circ(t,p,\pi)}
    {\bigl(\partial_\pi \Lambda^\circ(t,p,\pi)\bigr)^2}
    < 2
\]
for $\circ\in\{b,a\}$ and all $(t,p)\in[0,T]\times(0,1)$.
\end{enumerate}

\begin{remark}
    Although it is not needed mathematically, it is economically reasonable to choose $\Lambda^b$ and $\Lambda^a$ such that, for every $(t,p)\in[0,T]\times(0,1)$,
    \[
        \Lambda^b(t,p,0)=0
        \qquad
        \text{and}
        \qquad
        \Lambda^a(t,p,1)=0.
    \]
    This reflects the fact that trading activity vanishes at economically unreasonable price levels.\ Since the contract pays either \$0 or \$1, no rational market participant sells at a price of \$0 or buys at a price of \$1.
\end{remark}

\subsection{The Optimization Problem}
Let $\gamma>0$ denote the running risk aversion parameter and $\Phi\colon (0,1)\times \mathcal Q\to \R$ be a bounded continuous terminal penalty.\ The market maker chooses an admissible quoting strategy
$\pi\in\A$ to maximize
\begin{align*}
\E\bigg[
X_T + q_T Y +\Phi(p_T,q_T)
- \gamma\int_0^Tq_s^2 \varsigma(s,p_s)^2\,ds
\bigg],
\end{align*}
where $\A:=\A(0)$ and $\A(t)$ denotes the set of all predictable processes $\pi=(\pi^b_s,\pi^a_s)_{s\in[t,T]}$ taking values in \([0,1]^2\)
\(\P\otimes dt\)-almost everywhere for $t\in [0,T]$.

This objective balances expected terminal wealth against the risks associated with holding inventory. The term $X_T + q_T Y$ represents the market maker's terminal wealth, consisting of the terminal cash position and the settlement value of the remaining inventory. Prior to settlement, the mark-to-market value of the inventory is exposed to fluctuations in the market price.\ By penalizing its instantaneous variance rate $q_t^2 \varsigma(t,p_t)^2$, the running penalty discourages the market maker from maintaining large inventory positions, particularly when volatility is high.\

At settlement time $T$, any remaining inventory is additionally exposed to the binary settlement outcome.\ We account for this exposure through the terminal penalty $\Phi(p_T,q_T)$.\ A natural choice is to penalize the conditional variance 
$$\mathrm{Var}(q_T Y\mid \mathcal{F}_T)=q_T^2p_T(1-p_T)$$ of the settlement value at time $T$, that is,
\begin{equation}\label{eq:specialPhi}
\Phi(p,q):=-\gamma_T q^2 p(1-p),\qquad \text{for }p\in (0,1)\text{ and }q\in \mathcal Q,\tag{SV}
\end{equation}
with a terminal risk aversion parameter $\gamma_T>0$.\ Accordingly, the terminal penalty captures the settlement risk associated with inventory that has not been unwound prior to resolution.\ The parameters $\gamma$ and $\gamma_T$ allow the market maker to assign different weights to ongoing mark-to-market risk and terminal settlement risk.

Since $q_T$ is $\mathcal{F}_T$-measurable and $p_T=\E[Y\mid\mathcal{F}_T]$, we have $\E[q_TY]=\E[q_Tp_T]$.\
Consequently, the objective is equivalent to maximizing
\[
\E\!\l[
X_T + q_T p_T+\Phi(p_T,q_T)
- \gamma\int_0^T q_s^2 \varsigma(s,p_s)^2\,ds
\r]
\]
over all $\pi\in\mathcal{A}$.

\subsection{The Hamilton--Jacobi--Bellman Equation}
We now formulate the value function associated with the market maker's control problem and state the corresponding Hamilton--Jacobi--Bellman equation.\ We then reduce the four-dimensional equation to a three-dimensional one by exploiting the structure of the model and the objective functional.

The value function of the problem is given by
\begin{align}
\notag \Upsilon(t,p,q,x)
= \sup_{\pi\in\mathcal A(t)}
\E\Bigg[
X_T^{t,p,q,x,\pi} &+  q_T^{t,p,q,\pi} p_T^{t,p}
 +\Phi\big(p_T^{t,p},q_T^{t,p,q,\pi}\big)\\
&\qquad -\gamma\int_t^T\!\l(q_s^{t,p,q,\pi}\r)^2 \varsigma\big(s,p^{t,p}_s\big)^2\,ds\Bigg],\label{eq:def.value.function}
\end{align}
where $X_T^{t,p,q,x,\pi}$ denotes the cash position at time $T$ under the control $\pi$, starting from the initial state $(p,q,x)\in(0,1)\times\mathcal{Q}\times \R$ at time $t\in [0,T]$. The same notational convention is used for $q$.

The Hamilton--Jacobi--Bellman equation corresponding to our problem reads as
\begin{equation}\label{eq:HJB_pred_full}
\begin{aligned}
0
&= -\partial_t \Upsilon(t,p,q,x)
  - \frac12 \varsigma(t,p)^2\,\partial^2_{pp}\Upsilon(t,p,q,x)
  + \gamma q^2 \varsigma(t,p)^2 \\[1mm]
&\quad
  - \1_{\{q<Q\}}
    \sup_{\pi^b\in[0,1]}
      \Lambda^b(t,p,\pi^b)
      \l[\Upsilon(t,p,q+\Delta,x-\Delta\pi^b) - \Upsilon(t,p,q,x)\r] \\[1mm]
&\quad
  - \1_{\{q>-Q\}}
    \sup_{\pi^a\in[0,1]}
      \Lambda^a(t,p,\pi^a)
      \Big[\Upsilon(t,p,q-\Delta,x+\Delta\pi^a) - \Upsilon(t,p,q,x)\Big],
\end{aligned}
\end{equation}
for $(t,p,q,x)\in[0,T)\times (0,1)\times \mathcal Q\times \R$, with terminal condition
\[
    \Upsilon(T,p,q,x) = x + qp +\Phi(p,q).
\]

The cash position does not affect the dynamics of the price process or the inventory process, nor does it enter the order arrival intensities. Since changing the initial cash by a constant $c\in\R$ shifts the terminal cash position by the same constant, the value function satisfies
\[
    \Upsilon(t,p,q,x+c)=\Upsilon(t,p,q,x)+c
\]
for all $(t,p,q,x)\in[0,T]\times(0,1)\times\mathcal{Q}\times\R$ and $c\in\R$. Therefore, it suffices to consider the case of zero initial cash. We define the reduced value function $V\colon[0,T]\times(0,1)\times\mathcal{Q}\to\R$ by
\[
    V(t,p,q):=\Upsilon(t,p,q,0)-qp,
\]
so that the full value function can be recovered via
\begin{equation}\label{eq:value_reduction}
    \Upsilon(t,p,q,x)=x+qp+V(t,p,q).
\end{equation}
Substituting \eqref{eq:value_reduction} in the full HJB equation
\eqref{eq:HJB_pred_full} yields the reduced equation
\begin{align}
\notag
0
&= -\partial_t V(t,p,q)
  - \frac12 \varsigma(t,p)^2\,\partial^2_{pp}V(t,p,q)
  + \gamma q^2 \varsigma(t,p)^2 \\
\notag&\quad
  - \1_{\{q<Q\}} H^b\l(t,p;\frac{V(t,p,q) - V(t,p,q+\Delta)}{\Delta}\r) \\
&\quad
  - \1_{\{q>-Q\}} H^a\l(t,p;\frac{V(t,p,q) - V(t,p,q-\Delta)}{\Delta}\r) \tag{HJB}\label{eq:HJB_PDE_reduced_pred}
\end{align}
on $[0,T)\times(0,1)\times\mathcal{Q}$ with terminal condition
\begin{equation}\tag{TC}\label{eq:terminal_pred_reduced}
    V(T,p,q) = \Phi(p,q).
\end{equation}
Here, for $\circ\in\{b,a\}$, the function $H^\circ\colon[0,T]\times(0,1)\times\R\to\R$ is defined by
\begin{align*}
    H^\circ(t,p;z)
    &:= \Delta \sup_{\pi\in[0,1]}\Lambda^\circ\l(t,p,\pi\r)\l(G^\circ(p,\pi) - z\r)
\end{align*}
with
\[
    G^b(p,\pi) = p - \pi
    \qquad \text{and} \qquad
    G^a(p,\pi) = \pi - p.
\]

\begin{remark}
    Originally, the value function depends on the four variables $t$, $p$, $q$, and $x$.\ The representation above removes the dependence on the cash variable $x$, so that the dimensionality reduces from four to three.
\end{remark}

\section{Main Results}\label{sec:main_results}

We summarize the assumptions made so far.

\begin{assumption}[Global Assumptions]\label{ass.global}\
\begin{enumerate}
\item[(A1)] The transformation map $f\colon\R\to(0,1)$ is of class $C^2$, satisfies
\[
    \lim_{x\to-\infty} f(x)=0,
    \qquad
    \lim_{x\to\infty} f(x)=1,
\]
and is strictly increasing, i.e. $f'(x)>0$ for all $x\in\R$. Moreover,
\[
    g(x)=\frac{f''(x)}{f'(x)}
\]
is bounded and Lipschitz.

\item[(A2)] The volatility $\sigma\colon[0,T]\times\R\to(0,\infty)$ of the latent belief process is continuous, bounded, and Lipschitz continuous in the second argument, uniformly in time.

\item[(A3)] The functions
\[
    \Lambda^b,\Lambda^a:[0,T]\times(0,1)\times[0,1]\to[0,\infty),
\]
which determine the $\mathbb{F}$-predictable intensities
\[
    \lambda_t^b
    =
    \Lambda^b(t,p_t,\pi_t^b)\mathbf 1_{\{q_{t-}<Q\}}
    \qquad \text{and} \qquad
    \lambda_t^a
    =
    \Lambda^a(t,p_t,\pi_t^a)\mathbf 1_{\{q_{t-}>-Q\}},
\]
satisfy the following conditions:
\begin{enumerate}[(i)]
    \item $\Lambda^b$ and $\Lambda^a$ are continuous and uniformly bounded by some constant $\bar{\Lambda}>0$,
    \item $\Lambda^b$ and $\Lambda^a$ are twice continuously differentiable in $\pi$ on $(0,1)$,
    \item for every $(t,p,\pi)\in[0,T]\times(0,1)\times(0,1)$, we have
    \[
        \partial_\pi \Lambda^b(t,p,\pi)>0
        \qquad\text{and}\qquad
        \partial_\pi \Lambda^a(t,p,\pi)<0,
    \]
    \item $\Lambda^b$ and $\Lambda^a$ satisfy the curvature condition
    \[
        \sup_{\pi\in(0,1)}
        \frac{\Lambda^\circ(t,p,\pi)\,\partial_{\pi\pi}^2 \Lambda^\circ(t,p,\pi)}
        {\bigl(\partial_\pi \Lambda^\circ(t,p,\pi)\bigr)^2}
        < 2
    \]
    for $\circ\in\{b,a\}$ and all $(t,p)\in[0,T]\times(0,1)$.
\end{enumerate}
    \item[(A4)] The terminal penalty $\Phi\colon (0,1)\times \mathcal Q\to \R$ is bounded and continuous.
\end{enumerate}
\end{assumption}

\subsection{The Optimal Quoting Strategy}
We now turn to the optimization problems in the definitions of $H^b$ and $H^a$, since their maximizers determine the optimal bid and ask quotes. The following proposition establishes existence and uniqueness of these maximizers, characterizes them, and provides useful properties.

\begin{proposition}\label{lem:H_unique_maximizer_pred_model_general}
Assume that Assumption \ref{ass.global} is satisfied.\ Let $(t,p)\in[0,T]\times(0,1)$, $\circ\in\{b,a\}$, and define
\[
\Psi^\circ_{t,p}(\pi;z)
:=\Lambda^\circ(t,p,\pi)\bigl(G^\circ(p,\pi)-z\bigr)\qquad \text{for }z\in\R\text{ and }\pi\in[0,1].
\]
Moreover, for $\pi\in(0,1)$, let
\[
    u^b_{t,p}(\pi)
    :=G^b(p,\pi)-\frac{\Lambda^b(t,p,\pi)}{\partial_\pi \Lambda^b(t,p,\pi)},
    \qquad
    u^a_{t,p}(\pi)
    :=G^a(p,\pi)+\frac{\Lambda^a(t,p,\pi)}{\partial_\pi \Lambda^a(t,p,\pi)}.
\]
Then the following statements hold:
\begin{enumerate}[(i)]
\item For every $z\in\R$, there exists a unique maximizer $\pi^{\circ,*}_{t,p}(z)\in[0,1]$ of $\Psi^\circ_{t,p}(\,\cdot\,;z)$ given by
\[
\pi^{b,*}_{t,p}(z)=
\begin{cases}
0, & z\ge u^b_{t,p}(0+),\\[2mm]
\l(u^b_{t,p}\r)^{-1}(z),
& z\in u^b_{t,p}\big((0,1)\big),\\[2mm]
1, & z\le u^b_{t,p}(1-),
\end{cases}
\]
and
\[
\pi^{a,*}_{t,p}(z)=
\begin{cases}
0, & z\le u^a_{t,p}(0+),\\[2mm]
\l(u^a_{t,p}\r)^{-1}(z),
& z\in u^a_{t,p}\big((0,1)\big),\\[2mm]
1, & z\ge u^a_{t,p}(1-),
\end{cases}
\]
where
\[
    u^\circ_{t,p}(0+):=\lim_{\pi\downarrow 0}u^\circ_{t,p}(\pi)\in\overline{\R}
    \qquad\text{and}\qquad
    u^\circ_{t,p}(1-):=\lim_{\pi\uparrow 1}u^\circ_{t,p}(\pi)\in\overline{\R}
\]
with $\overline{\R}:=\R\cup\{-\infty\}$.

\item The map
$z\mapsto\pi^{\circ,*}_{t,p}(z)$ is of class $C^1$ on $J^\circ_{t,p}:=u^\circ_{t,p}\big((0,1)\big)$.\
Moreover, $\pi^{b,*}_{t,p}$ is strictly decreasing on $J^b_{t,p}$, whereas $\pi^{a,*}_{t,p}$ is strictly increasing on $J^a_{t,p}$.

\item The function $z\mapsto H^\circ(t,p;z)$ is decreasing on $\R$ and of class $C^2$ on $J^\circ_{t,p}$.

\item For every $z\in J^\circ_{t,p}$, we have
\[
\pi^{\circ,*}_{t,p}(z)
=\big(\Lambda^\circ(t,p,\,\cdot\,)\big)^{-1}
\bigg(-\frac{\partial_z H^\circ(t,p;z)}{\Delta}\bigg).
\]

\item The map
\[
    [0,T]\times(0,1)\times\mathbb{R}\to[0,1],
    \quad
    (t,p,z)\mapsto \pi_{t,p}^{\circ,*}(z)
\]
is continuous.
\end{enumerate}
\end{proposition}
The proof is deferred to Appendix~\ref{sec:proof.unique.maximizer}.

\subsection{Verification Theorem}
We now verify that the pointwise maximizers characterized in Proposition~\ref{lem:H_unique_maximizer_pred_model_general} indeed provide optimal
controls for the market making problem. To this end, we assume that the reduced Hamilton--Jacobi--Bellman equation \eqref{eq:HJB_PDE_reduced_pred} with terminal condition \eqref{eq:terminal_pred_reduced} admits a classical solution.\ In Theorem~\ref{thm.ex.unique.classical} below, we provide sufficient conditions for the existence and uniqueness of a classical solution, which are natural in view of the regularity results in \cite[Chapter 9]{Krylov1996}.\ The following theorem identifies the induced candidate with the value function of the control problem and proves optimality of the associated quoting strategy.

\begin{theorem}[Verification Theorem]
\label{thm:verification}
Assume that Assumption \ref{ass.global} is satisfied and that  there exists a classical solution 
$$V\in C^{1,2}\big([0,T)\times (0,1)\times \mathcal Q\big)\cap C_{\rm b}\big([0,T]\times (0,1)\times \mathcal Q\big)$$
to \eqref{eq:HJB_PDE_reduced_pred} with terminal condition \eqref{eq:terminal_pred_reduced}.\ Define
\[
    \Upsilon(t,p,q,x) := x+qp+V(t,p,q)
\]
for $(t,p,q,x)\in [0,T]\times(0,1)\times\mathcal Q\times\R$.\ Then $\Upsilon$ is the value function of the full control problem, namely
\[
\begin{aligned}
\Upsilon(t,p,q,x)
=
\sup_{\pi\in\mathcal A(t)}
\E\Bigg[
X_T^{t,p,q,x,\pi}
&+q_T^{t,p,q,\pi}p_T^{t,p}+\Phi\big(p_T^{t,p},q_T^{t,p,q,\pi}\big)\\
&\qquad-\gamma\int_t^T
\bigl(q_s^{t,p,q,\pi}\bigr)^2
\varsigma(s,p_s^{t,p})^2\,ds\Bigg],
\end{aligned}
\]
see \eqref{eq:def.value.function}.\ Moreover, an optimal control is given as follows.\ For $q\in\mathcal{Q}\setminus\{Q\}$, set
\[
        z_b(t,p,q)
        :=
        \frac{V(t,p,q)-V(t,p,q+\Delta)}{\Delta},
\]
and, for $q\in\mathcal{Q}\setminus\{-Q\}$, set
\[
        z_a(t,p,q)
        :=
        \frac{V(t,p,q)-V(t,p,q-\Delta)}{\Delta}.
\]
Define
\[
        \pi^{b,*}(t,p,q)
        \in
        \operatorname*{arg\,max}_{\pi\in[0,1]}
        \Lambda^b(t,p,\pi)
        \big(G^b(p,\pi)-z_b(t,p,q)\big),
        \quad \text{for } q\in\mathcal{Q}\setminus\{Q\},
\]
and
\[
        \pi^{a,*}(t,p,q)
        \in
        \operatorname*{arg\,max}_{\pi\in[0,1]}
        \Lambda^a(t,p,\pi)
        \big(G^a(p,\pi)-z_a(t,p,q)\big)
        \quad \text{for } q\in\mathcal{Q}\setminus\{-Q\}.
\]
At the boundary $q=Q$ and $q=-Q$, the bid and ask quotes may be chosen arbitrarily in $[0,1]$, respectively.\footnote{Economically, when $q=Q$, no bid quote is placed, and when $q=-Q$, no ask quote is placed.}\ Then the quoting strategy
\[
        \pi^*_s
        =
        \big(
        \pi^{b,*}(s,p_s,q_{s-}),
        \pi^{a,*}(s,p_s,q_{s-})
        \big),
        \qquad \text{for } s\in[t,T],
\]
is admissible and optimal.
\end{theorem}

The proof of Theorem~\ref{thm:verification} is contained in Appendix~\ref{app:proof.verification}.

Theorem~\ref{thm:verification} shows that, once a classical solution of the reduced Hamilton--Jacobi--Bellman equation is available, the optimal quoting strategy is obtained from the pointwise maximizers characterized in Proposition~\ref{lem:H_unique_maximizer_pred_model_general}.\ More precisely, by Theorem~\ref{thm:verification} and Proposition~\ref{lem:H_unique_maximizer_pred_model_general}, for every $(t,p,q)\in [0,T]\times(0,1)\times\mathcal{Q}$, the optimal bid and ask quotes are, in the setup of Proposition~\ref{lem:H_unique_maximizer_pred_model_general},
\[
    \pi_{t,p}^{b,*}\bigl(z_b(t,p,q)\bigr)
    \qquad\text{and}\qquad
    \pi_{t,p}^{a,*}\bigl(z_a(t,p,q)\bigr),
\]
where
\[
    z_b(t,p,q)=\frac{V(t,p,q)-V(t,p,q+\Delta)}{\Delta}
\]
and
\[
    z_a(t,p,q)=\frac{V(t,p,q)-V(t,p,q-\Delta)}{\Delta},
\]
whenever the neighboring inventory levels are admissible. At the inventory boundaries $q=Q$ and $q=-Q$, the quote on the constrained side is immaterial, since the corresponding order intensity is zero. In practice, this corresponds to withdrawing the quote on the constrained side rather than assigning it an arbitrary value.

It remains to prove existence and uniqueness of a classical solution to the reduced Hamilton--Jacobi--Bellman equation with the corresponding terminal condition, which is the purpose of the next subsection.

\subsection{Existence and Uniqueness of a Classical Solution}\label{sec:existence_uniqueness_mild_solution}
We now provide sufficient conditions in order to establish existence and uniqueness of a solution to \eqref{eq:HJB_PDE_reduced_pred} with terminal condition \eqref{eq:terminal_pred_reduced}.\ We adopt the notation from \cite[Chapter 8 and Chapter 9]{Krylov1996}.\ 
On $[0,\infty)\times \R$, we consider the parabolic distance $$d(z_1,z_2):= |t_1-t_2|^{1/2}+|x_1-x_2|$$
for $z_1=(t_1,x_1), z_2=(t_2,x_2)\in [0,\infty)\times \R$.\

Let $\alpha\in (0,1)$.\ For a nonempty set $D\subseteq \R$ or $D\subseteq [0,\infty)\times \R$, we identify functions on $D\times \mathcal Q$ with $\R^{\mathcal Q}$-valued functions defined on $D$, and define the H\"older space $$C^{\alpha}(D\times \mathcal Q):=C^{\alpha}(D;\R^{\mathcal Q})$$ or the parabolic H\"older space $$C^{\alpha/2,\alpha}(D\times \mathcal Q):=C^{\alpha/2,\alpha}(D;\R^{\mathcal Q})$$ as the set of all functions $u\in C_{\rm b}(D\times \mathcal Q)$ with
\[
 |u|_{\alpha}:=\|u\|_\infty+[u]_{\alpha}<\infty\quad \text{or}\quad |u|_{\alpha/2,\alpha}:=\|u\|_\infty+[u]_{\alpha/2,\alpha}<\infty,
\]
where
\begin{align*}
[u]_{\alpha}&:=\sup_{q\in \mathcal Q}\sup_{\substack{x_1,x_2\in D\\ x_1\neq x_2}}\frac{|u(x_1)-u(x_2)|}{|x_1-x_2|^\alpha}\quad \text{or}\\
[u]_{\alpha/2,\alpha}&:=\sup_{q\in \mathcal Q}\sup_{\substack{z_1,z_2\in D\\ z_1\neq z_2}}\frac{|u(z_1)-u(z_2)|}{d(z_1,z_2)^\alpha},
\end{align*}
respectively.\ The spaces $C^{2+\alpha}(D\times Q)$ and $C^{1+\alpha/2,2+\alpha}(D\times Q)$ are defined accordingly if $D$ is a domain.\ Moreover, $C^{\alpha}(D)$ and $C^{\alpha/2,\alpha}(D)$ denote the classical H\"older and parabolic H\"older spaces, see \cite[Definition 3.1.2, p.\ 34, and Section 8.5, p.\ 117]{Krylov1996} for the details, respectively.

Since the volatility coefficient $\varsigma(t,p)$, defined in \eqref{eq:def.varsigma}, may vanish as $p$ approaches $0$ or $1$, the corresponding second-order differential operator in $p$ needs not be uniformly elliptic.\ We therefore follow a classical approach and analyze a transformed version of the differential equation on $\R$, corresponding to a semilinear equation for the latent belief process $(L_t)_{t\in [0,T]}$.\ We work under the following stronger assumption.

\begin{assumption}[Existence and Uniqueness of a Classical Solution]\label{ass.exunique}
In addition to Assumption \ref{ass.global}, there exists $\alpha \in (0,1)$ such that
    \begin{enumerate}
\item[(i)] the coefficient $a:=\frac12\sigma^2$ satisfies $a\in C^{\alpha/2,\alpha}([0,T]\times \R)$ and $$\inf_{(t,x)\in [0,T]\times \R}a(t,x)>0,$$
\item[(ii)] the transformed intensities
\[
\Lambda_f^\circ\big(t,x,\pi\big)
:=
\Lambda^\circ(t,f(x),\pi),
\quad \text{for } (t,x)\in [0,T]\times \R, \; \pi\in [0,1],\; \circ\in\{b,a\},
\]
are uniformly parabolic $\alpha$-Hölder continuous in $(t,x)$, that is,
\[
\big[\Lambda_f^\circ\big]_{\alpha/2,\alpha}:=\sup_{\pi \in [0,1]} \big[\Lambda_f^\circ(\,\cdot\,, \pi)\big]_{\alpha/2,\alpha}<\infty \qquad \text{for }\circ\in \{b,a\},
\]
\item[(iii)] the map $$\varphi\colon \R\times \mathcal Q\to \R,\quad (x,q)\mapsto \Phi\big(f(x),q\big)$$ satisfies $\varphi\in C^{2+\alpha}(\R\times \mathcal Q)$. 
    \end{enumerate}
\end{assumption}

We say that a function $V$ is a classical solution to \eqref{eq:HJB_PDE_reduced_pred} with terminal condition \eqref{eq:terminal_pred_reduced} if
$$V\in C^{1,2}\big([0,T)\times (0,1)\times \mathcal Q\big)\cap C_{\rm b}\big([0,T]\times (0,1)\times \mathcal Q\big),$$
and $V$ satisfies \eqref{eq:HJB_PDE_reduced_pred} and \eqref{eq:terminal_pred_reduced} in a classical sense.

\begin{theorem}\label{thm.ex.unique.classical}
Assume that Assumption \ref{ass.exunique} is satisfied.\ Then, the nonlinear partial differential equation \eqref{eq:HJB_PDE_reduced_pred} together with terminal condition \eqref{eq:terminal_pred_reduced} admits a unique classical solution $V$.\ Moreover, $v(t,x,q):=V\big(t,f(x),q\big)$ for $t\in [0,T]$, $x\in \R$, and $q\in \mathcal Q$ satisfies
$$v\in C^{1+\alpha/2,2+\alpha}\big([0,T)\times \R\times \mathcal Q\big)\cap C^{\alpha/2,\alpha}\big([0,T]\times \R\times \mathcal Q\big).$$
\end{theorem}

A key issue is the nonlinear coupling of the equation across neighboring inventory levels $q\in\mathcal{Q}$.\ To address this, we adopt a fixed-point approach based on a representation as a mild solution.\ More precisely, we consider a frozen version of the equation, in which the occurrences of $V$ in $H^b$ and $H^a$ are replaced by a prescribed candidate function.\ This yields a linear equation which is used to define an operator that maps each candidate function to a mild solution of the corresponding frozen problem.\ We prove that this operator admits a unique fixed point, which is then shown to be the unique classical solution to \eqref{eq:HJB_PDE_reduced_pred} with terminal condition \eqref{eq:terminal_pred_reduced}.

In the sequel, we briefly outline the main steps of the proof of Theorem \ref{thm.ex.unique.classical}.\ The details of the proof are relegated to Appendix \ref{app:deferred_existence_and_uniqueness_of_a_mild_solution}.

We start by introducing the notion of a mild solution to a transformed version of the reduced Hamilton--Jacobi--Bellman equation \eqref{eq:HJB_PDE_reduced_pred} with terminal condition \eqref{eq:terminal_pred_reduced}.

For any function \(v\colon [0,T]\times \R\times Q\to\mathbb R\), define
\[
\begin{aligned}
\mathcal Hv(t,x,q)
&:=
\gamma q^2\varsigma\big(t,f(x)\big)^2 \\
&\quad
-\mathbf 1_{\{q<Q\}}
H^b\left(
t,f(x);
\frac{v(t,x,q)-v(t,x,q+\Delta)}{\Delta}
\right) \\
&\quad
-\mathbf 1_{\{q>-Q\}}
H^a\left(
t,f(x);
\frac{v(t,x,q)-v(t,x,q-\Delta)}{\Delta}
\right).
\end{aligned}
\]

For $t\in[0,T]$ and every twice continuously differentiable function $u\colon\R\to\R$, let
\[
    (\mathcal L_tu)(x):=a(t,x)\bigg(\dv[2]{}{x}u(x)-g(x)\dv{}{x}u(x)\bigg) \qquad\text{for } x\in\R.
\]
With this notation, for arbitrary but fixed $\Phi\in C_{\rm b}\big((0,1)\times \mathcal Q\big)$, a function $V$ is a classical solution to the reduced Hamilton--Jacobi--Bellman equation \eqref{eq:HJB_PDE_reduced_pred} with terminal condition \eqref{eq:terminal_pred_reduced} if and only if $v(t,x,q):=V\big(t,f(x),q\big)$ is a classical solution to the abstract Cauchy problem 
\begin{equation}
   \begin{aligned} \partial_t v(t,x,q) + \bigl(\mathcal L_tv(t,\, \cdot\,, q)\bigr)(x) &= \mathcal Hv(t,x,q),\\
    v(T,x,q) &= \Phi\big(f(x),q\big), 
    \end{aligned}
    \tag{CP}\label{eq:cp}
\end{equation}
for $(t,x,q)\in[0,T)\times\R\times \mathcal Q$.\

A bounded continuous function
\[
    v\colon[0,T]\times\R\times\mathcal Q \to \R
\]
is called a mild solution to \eqref{eq:cp} if, for all $(t,x,q)\in[0,T]\times\R\times\mathcal Q$,
\begin{equation}\label{eq:def.mild.solution}
    v(t,x,q)
    =
    \E\big[\varphi(L_T^{t,x},q)\big]-
    \int_t^T
    \E\big[(\mathcal Hv)(s,L_s^{t,x},q)\big]\,ds\tag{Mild}
\end{equation}
with $\varphi(x,q):=\Phi\big(f(x),q\big)$.

In the proof of Theorem \ref{thm:existence_uniqueness_mild}, we employ a fixed point argument, based on Banach's fixed point theorem in the space $C_{\rm b}\big([0,T]\times \R\times \mathcal Q\big)$ to obtain the existence and uniqueness of a mild solution, i.e., a function $v\in C_{\rm b}\big([0,T]\times \R\times \mathcal Q\big)$ satisfying \eqref{eq:def.mild.solution}, which is then proved to be the unique classical solution to \eqref{eq:cp} in Section \ref{app:proof.classical}.

The fixed-point construction employed in the proof of Theorem \ref{thm:existence_uniqueness_mild} is also of computational interest, as it suggests a natural iterative procedure for approximating the unique classical solution.

\begin{remark} Let $v^{(0)}\in\mathcal U$ and define, recursively,
\[
    v^{(n)}:=\Gamma v^{(n-1)}
    \qquad \text{for }n\in\mathbb N
\]
with $\mathcal U$ and $\Gamma$ as in the proof of Theorem \ref{thm:existence_uniqueness_mild}.\ If we choose $\beta >0$ sufficiently large as in the proof of Theorem \ref{thm:existence_uniqueness_mild}, the operator $\Gamma$ is a contraction on $(\mathcal U,\|\cdot\|_\beta)$ with contraction constant
\[
    \kappa = \frac{4\overline\Lambda}{\beta} \in(0,1),
\]
see Appendix \ref{app:deferred_existence_and_uniqueness_of_a_mild_solution.mild} for the details.\ We then obtain
\[
    \|v^{(n)}-v\|_\beta
    \le
    \kappa \|v^{(n-1)}-v\|_\beta
    \le
    \kappa^n\|v^{(0)}-v\|_\beta
\]
for every $n\in\N$, where $v\in\mathcal U$ is the unique fixed point of $\Gamma$. In particular,
\[
    \|v^{(n)}-v\|_\infty \le e^{\beta T}\|v^{(n)}-v\|_\beta \to 0\qquad \text{as }n\to\infty.
\]
\end{remark}

We point out that this iterative scheme, based on the fixed-point argument, yields a natural numerical approximation method with explicit convergence guarantees due to Banach's fixed point theorem.\ In the next section, we present an alternative numerical approach that uses a finite difference scheme for the reduced Hamilton--Jacobi--Bellman equation.

\section{Numerical Analysis of the Optimal Quoting Strategy}\label{sec:numerical_analysis_optimal_quotes}
We now study the optimal quoting strategy numerically.\ To this end, we specify parameters within the framework of Section~\ref{sec:model} and approximate the reduced value function by solving \eqref{eq:HJB_PDE_reduced_pred} using a finite difference scheme. The resulting solution yields optimal bid and ask quotes as functions of time, inventory, and price.\ We then analyze their spread and skew and compare the performance against a myopic baseline strategy that ignores inventory risk.

\subsection{Specification of the Model}
We map latent beliefs to prices using the logistic function
\[
    f\colon\R\to(0,1),\quad x\mapsto\frac{1}{1+e^{-x}}.
\]
The price process therefore satisfies
\[
    dp_t = p_t(1-p_t)\tilde\sigma(t,p_t)\, dW_t,
\]
where $\tilde\sigma(t,p) := \sigma(t,f^{-1}(p))$. We specify
\[
    \sigma(t,x)
    =
    \sigma_0
    +
    \sigma_1 \left(\frac{t}{T}\right)^\eta
    +
    \frac{\sigma_2}{1+x^2}
    \qquad \text{for }(t,x)\in[0,T]\times\R,
\]
where $\sigma_0,\sigma_1,\sigma_2>0$ and $\eta\ge 1$.\ Hence, for $t\in[0,T]$ and $p\in(0,1)$, the function $\tilde\sigma$ is given by
\[
    \tilde\sigma(t,p)
    =
    \sigma_0
    +
    \sigma_1 \left(\frac{t}{T}\right)^\eta
    +
    \frac{\sigma_2}{1+\left(\ln\frac{p}{1-p}\right)^2}.
\]
The specification combines (i) a baseline volatility $\sigma_0$, (ii) a time-increasing component $\sigma_1 \left(\nicefrac{t}{T}\right)^\eta$ capturing faster information flow near settlement, and (iii) an uncertainty term $\nicefrac{\sigma_2}{(1+x^2)}$, which is largest at $p=\nicefrac{1}{2}$ and decays as beliefs become more extreme. Hence volatility increases as resolution approaches and is amplified when the event outcome remains most uncertain.

To the best of our knowledge, the empirical functional form of order intensities in prediction markets has not yet been studied systematically.\ The following specification should therefore be viewed as a tractable modeling choice that captures reasonable assumptions on the intensity. For $(t,p,\pi)\in[0,T]\times(0,1)\times[0,1]$, we define
\[
    \Lambda^\circ(t,p,\pi) := A(t,p)B^\circ(t,p,\pi) \qquad\text{for } \circ\in\{b,a\},
\]
where $A(t,p)$ captures overall market activity and $B^\circ(t,p,\pi)$ describes the dependence on the quoted price. We set
\[
    A(t,p)
    :=
    \l(
    A_0
    +
    A_1 \frac{e^{\xi t/T}-1}{e^\xi-1}
    \r)
    \sqrt{p(1-p)},
\]
with \(\xi,A_0,A_1>0\). Hence activity increases toward resolution and is maximal when uncertainty is highest, i.e., when $p=\nicefrac{1}{2}$.

The bid and ask shape functions are specified as
\[
    B^{b}(t,p,\pi)
    :=
    \l(\frac{2\pi}{\pi+p}\r)^{\nu}
    \exp\l(-k(t)(p-\pi)\r)
\]
and
\[
    B^{a}(t,p,\pi)
    :=
    \l(\frac{2(1-\pi)}{2-\pi-p}\r)^{\nu}
    \exp\l(-k(t)(\pi-p)\r),
\]
where $\nu>0$. The time-dependent liquidity parameter is given by
\[
    k(t) := k_0 + (k_1-k_0)\frac{e^{\nicefrac{\kappa t}{T}}-1}{e^\kappa-1} 
\]
with $\kappa,k_0,k_1>0$. This specification builds on exponential execution intensities in the bid and ask spreads, $\delta^b=p-\pi^b$ and $\delta^a=\pi^a-p$, while incorporating two prediction-market features. First, the multiplicative factors ensure that bid intensities vanish as quotes approach zero and ask intensities vanish as quotes approach one.\ They also reflect that a fixed absolute spread corresponds to a larger relative belief deviation near the boundaries than near the center of the price interval.\ The parameter $\nu$ controls the strength of this effect.\ Second, $k(t)$ allows spread sensitivity to vary over time.\ If $k_1>k_0$, liquidity increases toward settlement, so execution intensities decay more rapidly with the spread and quotes must be placed closer to the current price to achieve a given fill rate.

We set $T=1$ and use the values in Table~\ref{tab:parameters} unless stated otherwise.

\begin{table}[htbp]
\centering
\small
\setlength{\tabcolsep}{5pt}
\renewcommand{\arraystretch}{0.95}
\begin{tabular}{@{}llp{0.55\linewidth}@{}}
\toprule
\textbf{Parameter} & \textbf{Value} & \textbf{Description} \\
\midrule
\multicolumn{3}{@{}l}{\textit{Volatility}} \\
$\sigma_0$ & $0.6$ & Baseline level \\
$\sigma_1$ & $0.3$ & Time-acceleration scale \\
$\sigma_2$ & $0.1$ & Uncertainty scale \\
$\eta$     & $3$   & Acceleration exponent \\[3pt]

\multicolumn{3}{@{}l}{\textit{Trading activity}} \\
$A_0$   & $100$ & Baseline level \\
$A_1$   & $150$ & Time-variation scale \\
$\xi$ & $2$   & Acceleration exponent \\[3pt]

\multicolumn{3}{@{}l}{\textit{Intensity shape}} \\
$\nu$  & $1$   & Boundary decay \\
$k_0$    & $35$  & Initial spread sensitivity \\
$k_1$    & $50$  & Terminal spread sensitivity \\
$\kappa$ & $1.5$ & Sensitivity acceleration \\[3pt]

\multicolumn{3}{@{}l}{\textit{Risk and inventory}} \\
$\gamma$   & $4\cdot10^{-3}$ & Running inventory penalty \\
$\gamma_T$ & $10^{-3}$       & Terminal inventory penalty \\
$Q$        & $100$           & Inventory limit \\
$\Delta$   & $10$            & Trade size \\
\bottomrule
\end{tabular}
\caption{Model parameters.}
\label{tab:parameters}
\end{table}

The parameters are chosen to illustrate qualitative effects rather than to calibrate the model to a specific contract.

\subsection{Numerical Scheme}
The reduced value function $V(t,p,q)$ is computed numerically by solving the reduced Hamilton--Jacobi--Bellman equation on a grid in $(t,p,q)$. Time and price are discretized uniformly on $[0,T]$ and $[p_{\min},p_{\max}]$, respectively, with $p_{\min}=10^{-3}$ and $p_{\max}=1-10^{-3}$. Inventory takes values on the finite grid $\mathcal Q$.

Let $\{t_n\}_{n=0}^{N_t}$ denote the time grid with $t_0=0$ and $t_{N_t}=T$, let $h := t_1-t_0$ denote its step size, and let $\{p_j\}_{j=0}^{N_p}$ denote the uniform $p$-grid. The equation is solved backward in time using the implicit Euler method with terminal condition \eqref{eq:specialPhi}, i.e.,
\[
    V(t_{N_t},p,q) = -\gamma_T q^2 p(1-p).
\]
 The second derivative $\partial_{pp}^2 V$ is approximated by a second-order finite difference scheme. At the truncated boundaries $p_{\min}$ and $p_{\max}$, homogeneous Neumann boundary conditions are imposed. This yields a sparse matrix $D_2$ such that
\[
    D_2 u \approx \partial_{pp}^2 u
\]
for functions $u$ defined on the $p$-grid. For fixed $n\in\{0,\ldots,N_t\}$ and $q\in\mathcal{Q}$, define the vectors $\varsigma_n$ and $V_n(q)$ by
\[
    \varsigma_{n,j} := p_j(1-p_j)\,\tilde\sigma(t_n,p_j), \qquad \text{for }j\in\{0,\ldots,N_p\},
\]
and
\[
    V_{n,j}(q) := V(t_n,p_j,q) \qquad \text{for }j\in\{0,\ldots,N_p\}.
\]
Moreover, the vector $\varsigma_n^2$ is understood pointwise, that is,
\[
    (\varsigma_n^2)_j = \varsigma_{n,j}^2
\]
for $j\in\{0,\ldots,N_p\}$.

Our discretization of \eqref{eq:HJB_PDE_reduced_pred} at time $t_n$ yields
\begin{align}\label{eq:fully_discrete_hjb}
    A_n V_n(q)
    =
    V_{n+1}(q)
    -
    h\Bigl(
    \gamma q^2 \varsigma_n^2
    -
    \mathcal H^b_n(q)
    -
    \mathcal H^a_n(q)
    \Bigr)
\end{align}
for $n\in\{0,\ldots,N_t-1\}$ and $q\in\mathcal{Q}$, where
\[
    A_n := I - \frac{h}{2}\,\operatorname{diag}(\varsigma_n^2)\,D_2,
\]
$I := I_{N_p+1}$ denotes the identity on the $p$-grid, and the vectors $\mathcal H^b_n(q)$ and $\mathcal H^a_n(q)$ are defined by
\[
    \mathcal{H}^b_{n,j}(q)
    :=
    \1_{\{q<Q\}}
    H^b\!\l(
    t_n,
    p_j,
    \frac{V_{n,j}(q)-V_{n,j}(q+\Delta)}{\Delta}
    \r)
\]
and
\[
    \mathcal{H}^a_{n,j}(q)
    :=
    \1_{\{q>-Q\}}\,
    H^a\!\l(
    t_n,
    p_j,
    \frac{V_{n,j}(q)-V_{n,j}(q-\Delta)}{\Delta}
    \r)
\]
for $j\in\{0,\ldots,N_p\}$.

Equation \eqref{eq:fully_discrete_hjb} defines a nonlinear system for the unknown family
\[
    V_n := \{ V_n(q) \}_{q \in \mathcal Q}.
\]
We solve this system by a fixed-point iteration, starting from the initial guess $V_n^{(0)}=V_{n+1}$. Given an iterate $V_n^{(k)}$, we compute
\[
\mathcal H^{b,(k)}_{n,j}(q)
=
\1_{\{q<Q\}}
H^b\l(
t_n,
p_j,
\frac{V_{n,j}^{(k)}(q)-V_{n,j}^{(k)}(q+\Delta)}{\Delta}
\r)
\]
and
\[
\mathcal H^{a,(k)}_{n,j}(q)
=
\1_{\{q>-Q\}}
H^a\!\l(
t_n,
p_j,
\frac{V_{n,j}^{(k)}(q)-V_{n,j}^{(k)}(q-\Delta)}{\Delta}
\r)
\]
for $q\in\mathcal{Q}$ and $j\in\{0,\ldots,N_p\}$ by maximization over a fine uniform grid of candidate quotes. The fixed-point update is obtained by solving
\[
    A_n V_n^{(k+1)}(q)
    =
    V_{n+1}(q)
    -
    h\l(
    \gamma q^2 \varsigma_n^2
    -
    \mathcal H^{b,(k)}_n(q)
    -
    \mathcal H^{a,(k)}_n(q)
    \r)
\]
for every $q\in\mathcal Q$. We solve this system using an LU decomposition of $A_n$. Since $A_n$ does not depend on the iteration index $k$, the decomposition is computed once per time step and reused throughout the fixed-point iteration. The iteration is terminated once
\[
    \big\|V_n^{(k+1)} - V_n^{(k)}\big\|_\infty
\]
falls below a predetermined tolerance or a maximum number of iterations is reached.

After computing $V$ on the grid, the optimal quotes for arbitrary $(t,p,q)$ are obtained by trilinear interpolation of $V$ in $(t,p,q)$, followed by the same maximization method used above.\ We point out that the trilinear interpolation is carried out to be able to handle $t$ and $p$ values, which are not on the numerical grid, and to extend the model from the discrete set of possible inventories $\mathcal{Q}$ to allow for arbitrary inventory levels within the prescribed boundaries.

\subsection{Analysis of the Optimal Quoting Strategy}\label{sec:analysis_optimal_strategy}
We now analyze the optimal bid and ask quotes obtained from the numerical solution of the HJB equation. The analysis proceeds in two steps.\ First, we examine the bid-ask spread and skew under the parameter setting reported in Table~\ref{tab:parameters}.\ Second, we vary the risk aversion parameters in order to study how risk preferences alter the optimal quotes.

\subsubsection{Spread and Skew}
For a state $(t,p,q)\in[0,T]\times(0,1)\times\mathcal{Q}$, define the bid-ask spread and skew by
\[
    \mathrm{Spread}(t,p,q)
    :=
    \pi^a(t,p,q)-\pi^b(t,p,q),
\]
and
\[
    \mathrm{Skew}(t,p,q)
    :=
    \frac{\pi^a(t,p,q)+\pi^b(t,p,q)}{2}-p,
\]
respectively. The bid-ask spread measures the compensation required for providing liquidity, while the skew measures the displacement of the quote midpoint from the current price.

Figure~\ref{fig:spread_vs_time_q0} shows the spread at zero inventory for different price levels as a function of time. Since the bid-ask spread is symmetric around $p=\nicefrac{1}{2}$, we restrict our attention to $p\le\nicefrac{1}{2}$.
\begin{figure}[htbp]
    \centering
    \includegraphics[width=0.75\linewidth]{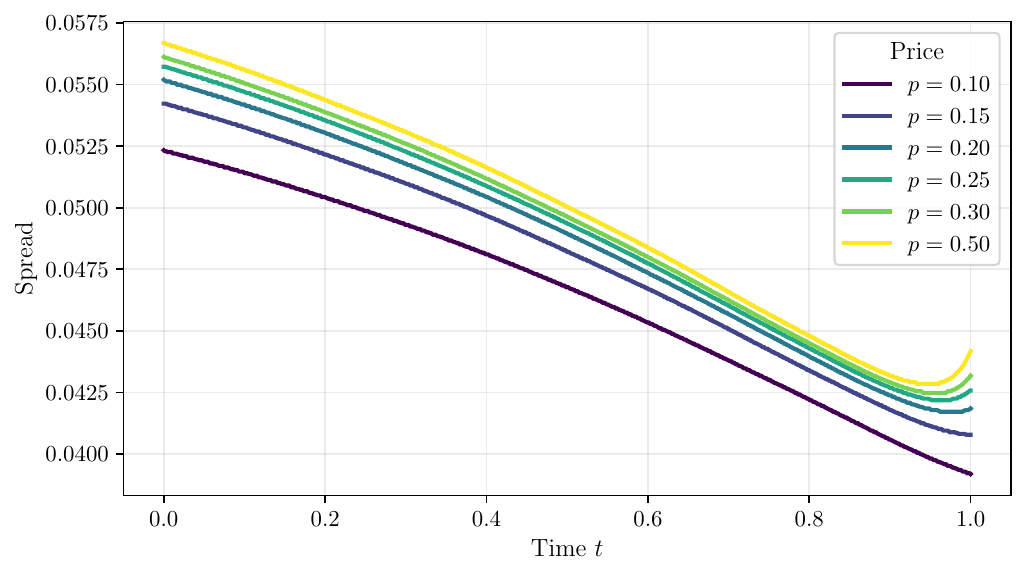}
    \caption{Bid-ask spread of the optimal quotes at flat inventory as a function of time $t$ for several price levels $p$.}
    \label{fig:spread_vs_time_q0}
\end{figure}
The spread generally decreases over time, as the liquidity parameter $k(t)$ rises, which makes execution intensities more spread-sensitive and encourages quoting closer to the current price. Near settlement, however, the spread widens for prices close to $p=\nicefrac{1}{2}$, where settlement risk is highest and there is little time left to unwind new positions. For the same reason, for any fixed time, spreads are larger for prices closer to $p=\nicefrac{1}{2}$.

We next consider the skew. Due to the multiplicative factors in the shape functions, symmetric quotes around the current price generally do not imply symmetric execution intensities.\ For $p<\nicefrac{1}{2}$, the bid intensity is lower than the ask intensity at equal spreads, so the bid is optimally placed closer to $p$.\ For $p>\nicefrac{1}{2}$, the reverse holds.\ Thus, even without inventory pressure, profit maximization induces a positive skew for $p<\nicefrac{1}{2}$ and a negative skew for $p>\nicefrac{1}{2}$.\ This effect becomes more pronounced as the price approaches zero or one.

Figure~\ref{fig:skew_multiple_q} shows the skew as a function of price for different inventory levels. At zero inventory, the skew is positive for prices below $p=\nicefrac{1}{2}$ and negative for prices above it.\ Inventory shifts the skew in the expected direction as short positions move quotes upward, whereas long positions move them downward. Near the boundaries, the curves converge because risk vanishes as $p(1-p)\to0$.

\begin{figure}[htbp]
    \centering
    \includegraphics[width=0.75\linewidth]{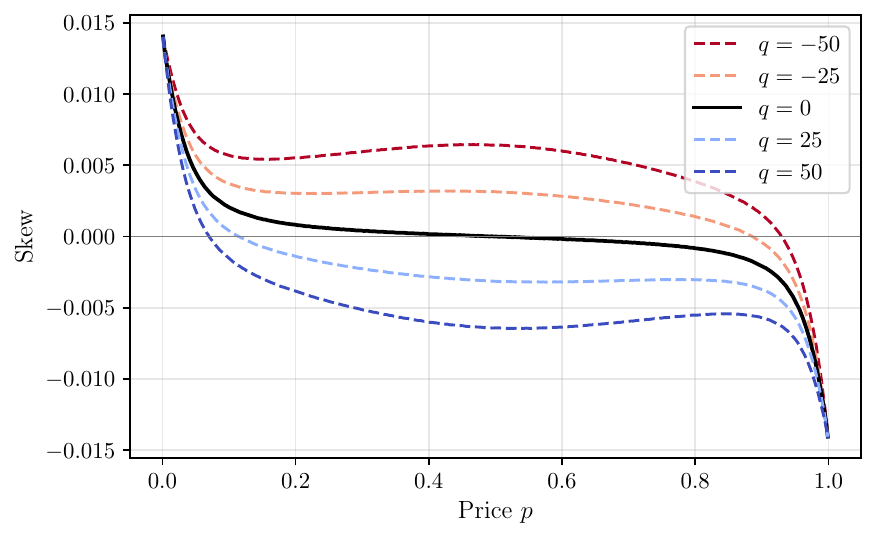}
    \caption{Skew of the optimal quotes at time $t=0$ as a function of price $p$ for inventory levels $q\in\{-50,-25,0,25,50\}$.}
    \label{fig:skew_multiple_q}
\end{figure}

To isolate inventory control, Figure~\ref{fig:heatmap_skew_qt_p0p50} fixes $p=\nicefrac{1}{2}$. The skew is decreasing in inventory and vanishes at $q=0$. Short positions induce positive skew to encourage buy executions and discourage further selling. Conversely, long positions induce negative skew to encourage sell executions and discourage further buying. Over time, the skew reflects two opposing effects. As the remaining horizon shortens, the running inventory penalty becomes less important, which pushes the skew toward zero. Closer to settlement, however, there is less time to unwind positions before the terminal penalty is imposed, which leads to stronger quote adjustments.

\begin{figure}[htbp]
    \centering
    \includegraphics[width=0.75\linewidth]{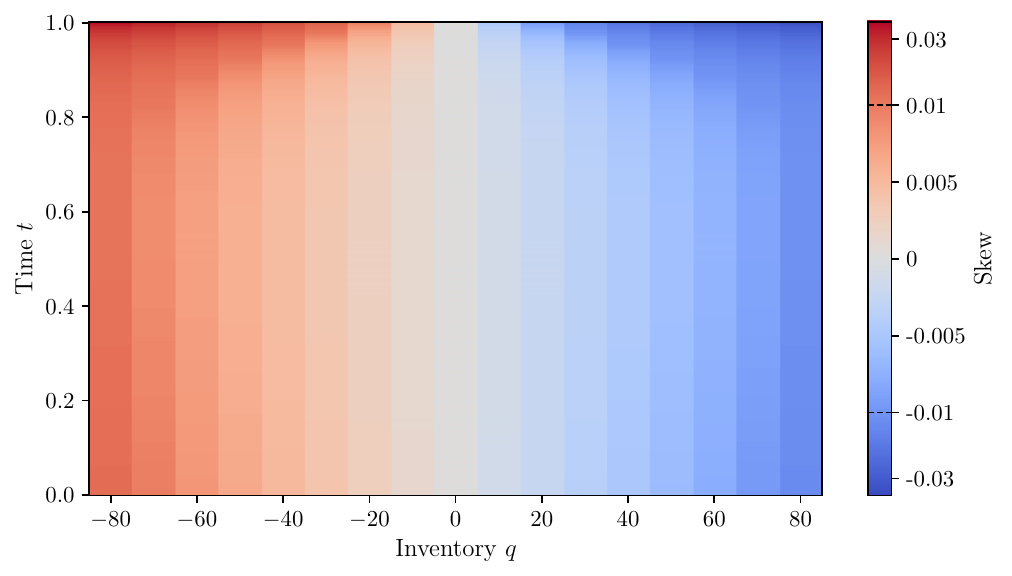}
    \caption{Skew of the optimal quotes at $p=\nicefrac{1}{2}$ as a function of time $t$ and inventory $q$.\ The color scale uses a symmetric logarithmic normalization with a linear region around zero up to a threshold of $0.01$.\ The dashed lines in the colorbar mark the transition between the linear and logarithmic regions.}
    \label{fig:heatmap_skew_qt_p0p50}
\end{figure}

\subsubsection{Risk Aversion Sensitivity}
We finally study how the skew changes with running and terminal risk aversion. We consider
\[
    \gamma_{\mathrm{low}} := 10^{-3}, \quad \gamma_{\mathrm{high}} := 10^{-2}, \quad \gamma_{T,\mathrm{low}} := 2\cdot10^{-4}, \quad \gamma_{T,\mathrm{high}} := 2\cdot10^{-3}.
\]
Figure~\ref{fig:skew_2x2_qt_p0p50} shows the skew at $p=\nicefrac{1}{2}$ for the four combinations of $(\gamma,\gamma_T)$.

\begin{figure}[htbp]
    \centering
    \includegraphics[width=0.75\linewidth]{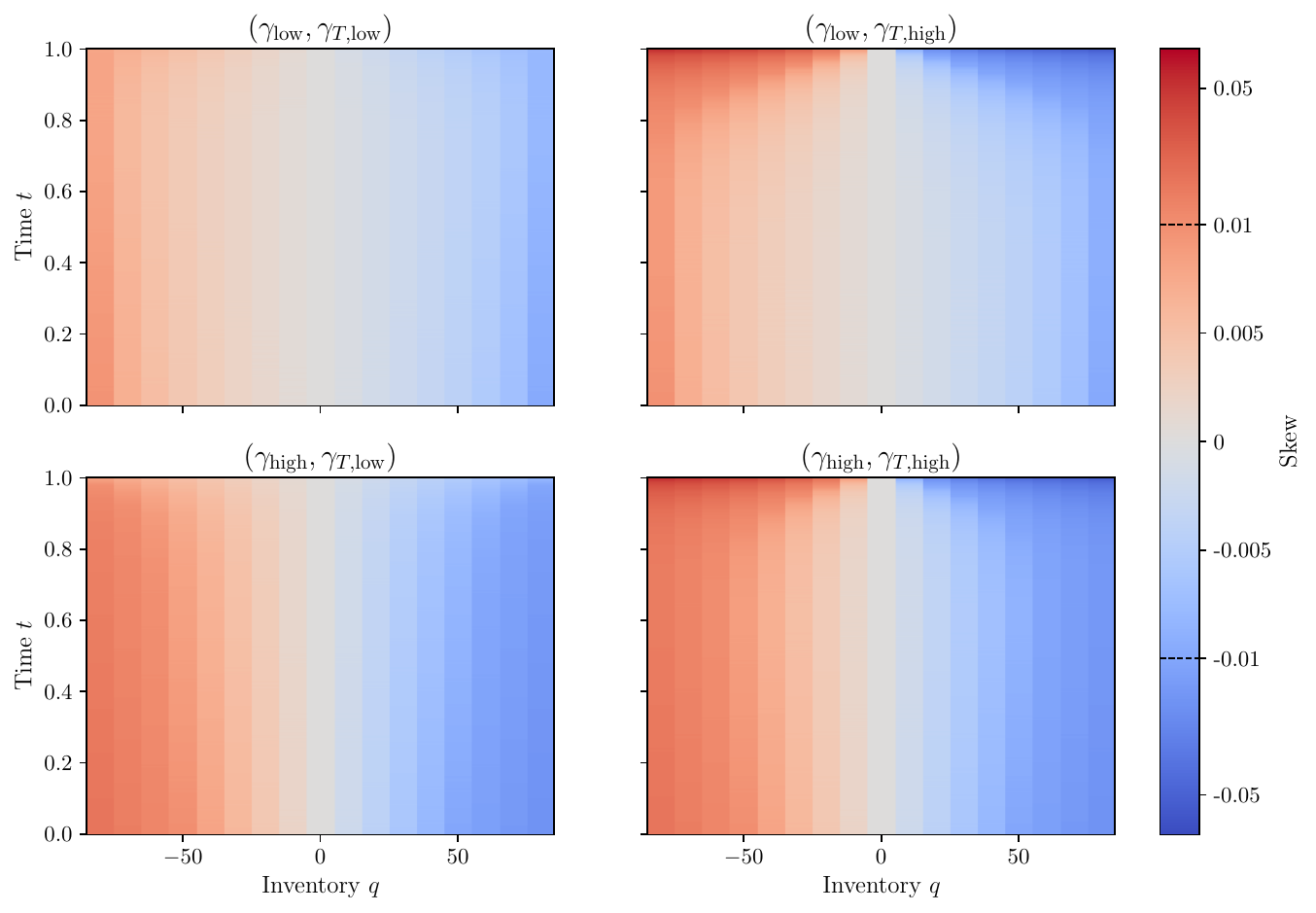}
    \caption{Skew of the optimal quotes at $p=\nicefrac{1}{2}$ as a function of time $t$ and inventory $q$ for different values of $(\gamma,\gamma_T)$. The color scale uses symmetric logarithmic normalization with a linear region around zero up to a threshold of $0.01$. The dashed lines in the colorbar mark the transition between the linear and logarithmic regions.}
    \label{fig:skew_2x2_qt_p0p50}
\end{figure}

Across all parameter combinations, the skew is decreasing in inventory and, as before, is zero when inventory is flat. Increasing the running risk aversion $\gamma$ mainly affects early times, when inventory is costly over a longer remaining horizon. Its effect weakens as settlement approaches. By contrast, increasing the terminal risk aversion $\gamma_T$ has the strongest effect near settlement, where inventory must be reduced over a short time interval, but its influence also propagates backward in time.

\subsection{Monte Carlo Simulation}
We compare the optimal quoting strategy with a myopic benchmark that maximizes instantaneous expected mark-to-market profit and ignores inventory risk. Both strategies are evaluated on the same simulated price paths and payoffs.

We simulate the price processes starting from $p_0 = \nicefrac{1}{2}$ on a uniform grid using the Euler--Maruyama method and project prices onto $[10^{-3},1-10^{-3}]$. At each time step, bid and ask market order arrivals are sampled independently from Poisson distributions whose parameters are determined by the current state, quotes, and inventory constraints.\footnote{Only the optimal strategy is subject to the inventory constraint, whereas the myopic benchmark is simulated without an inventory bound. In the reported specification, however, the benchmark exceeds the corresponding inventory bound only very rarely. Imposing the bound would require the benchmark to withdraw the quote on the constrained side, which would reduce its expected PnL.} Inventory and cash are then updated according to the executed buy and sell orders. At maturity, the contract payoff is sampled as
\[
    Y \sim \mathrm{Bernoulli}(p_T).
\]
Since $X_0=q_0=0$, the terminal profit and loss is
\[
    \mathrm{PnL}=X_T+q_TY.
\]

The results, based on $10\,000$ Monte Carlo paths, are reported in Table~\ref{tab:performance_metrics}.
\begin{table}[htbp]
\centering
\caption{Performance comparison of the baseline and optimal quoting strategies based on $10\,000$ simulated paths. The table reports the mean and standard deviation of PnL, the average absolute terminal inventory $|q_T|$, and the 5\% value at risk $\mathrm{VaR}_{5\%}$ and expected shortfall $\mathrm{ES}_{5\%}$.}
\label{tab:performance_metrics}
\begin{tabular}{lrrrrr}
\toprule
Strategy & $\mathrm{PnL}$ & $\mathrm{std}(\mathrm{PnL})$ & $|q_T|$ & $\mathrm{VaR}_{5\%}$ & $\mathrm{ES}_{5\%}$ \\
\midrule
Baseline & 12.47 & 28.11 & 49.37 & 32.41 & 40.50 \\
Optimal & 12.39 & 10.34 & 15.23 & 4.20 & 9.68 \\
\bottomrule
\end{tabular}
\end{table}
The optimal strategy attains nearly the same mean PnL as the baseline strategy, but with substantially lower risk.\ In particular, it reduces the standard deviation of PnL from $28.11$ to $10.34$, the average absolute terminal inventory from $49.37$ to $15.23$, the $5\%$ value at risk from $32.41$ to $4.20$, and the $5\%$ expected shortfall from $40.50$ to $9.68$.

These findings show that the optimal strategy achieves a large reduction in inventory and downside risk at only a small cost in expected profit.
\begin{figure}[htbp]
    \centering
    \begin{subfigure}{0.48\linewidth}
        \centering
        \includegraphics[width=\linewidth]{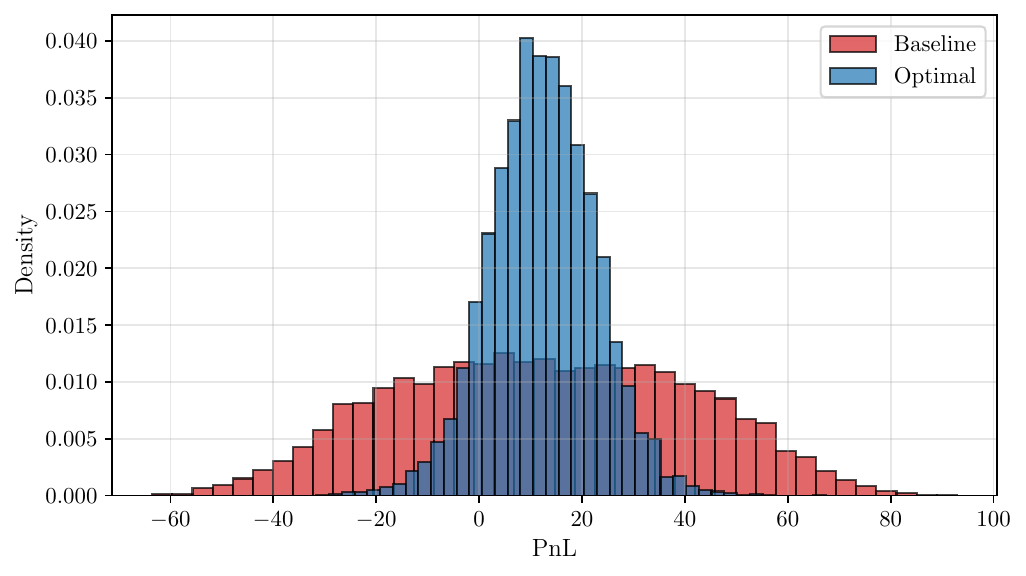}
        \caption{Final profit and loss.}
        \label{fig:pnl_distribution_baseline_our_strategy}
    \end{subfigure}
    \hfill
    \begin{subfigure}{0.48\linewidth}
        \centering
        \includegraphics[width=\linewidth]{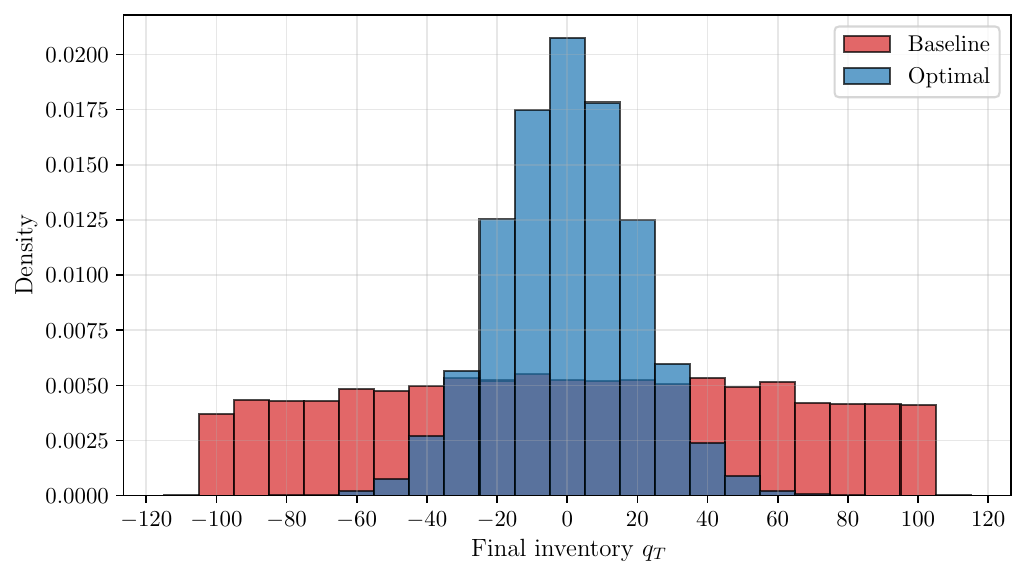}
        \caption{Terminal inventory.}
        \label{fig:inventory_distribution_baseline_our_strategy}
    \end{subfigure}
    \caption{Distributions of final profit and loss and terminal inventory for the baseline and optimal strategies based on $10\,000$ simulated paths.}
    \label{fig:pnl_inventory_distributions}
\end{figure}
Figures~\ref{fig:pnl_distribution_baseline_our_strategy} and~\ref{fig:inventory_distribution_baseline_our_strategy} confirm this pattern. Relative to the baseline, the optimal strategy produces a more concentrated PnL distribution and keeps terminal inventory much closer to zero.

\section{Conclusion}\label{sec:conclusion}
We developed a stochastic control framework for market making in prediction markets and derived the associated optimal quoting strategy.\ The model reflects key features of prediction market contracts.\ Prices are interpreted as conditional probabilities and therefore take values in $(0,1)$, contracts settle at a fixed terminal time according to a binary outcome, and admissible quotes are constrained to the interval $[0,1]$.\ The market maker maximizes expected terminal wealth while penalizing both running inventory exposure and terminal settlement risk.

The resulting Hamilton--Jacobi--Bellman equation was reduced from four to three dimensions.\ We then established existence and uniqueness of a classical solution to the reduced equation and characterized the optimal bid and ask quotes.

The numerical analysis illustrates the dependence of the optimal quotes on inventory, time to settlement, price, and risk aversion. In particular, it shows how asymmetric order arrival intensities generate skew and how inventory risk becomes less relevant as prices approach zero or one.\ The simulation study further demonstrates the risk-reduction effect of the optimal quoting strategy.\ Relative to a myopic benchmark that maximizes instantaneous expected mark-to-market profit, the optimal strategy substantially reduces downside risk at only a small cost in expected profit.

\appendix 
\section{Auxiliary Results}\label{app:deferred_market_model_and_control_problem}
The following lemma establishes boundedness and Lipschitz continuity of the first and second derivatives of the transformation function $f$ under the assumption that $g=\frac{f''}{f'}$ is bounded and Lipschitz continuous.
\begin{lemma}\label{lem:f'_bounded}
    Let $f\in C^2\big(\R;(0,1)\big)$ satisfy
    \[
        \lim_{x\to-\infty} f(x)=0, 
        \qquad 
        \lim_{x\to\infty} f(x)=1,
    \]
    and $f'(x)>0$ for all $x\in\R$.\ If the function $g\colon\R\to\R$, defined by
    \[
        g(x) := \frac{f''(x)}{f'(x)}\qquad\text{for all }x\in \R,
    \]
    is bounded, then $f'$ is bounded and Lipschitz.\ If $g$ is, in addition, Lipschitz, it follows that $f''$ is also bounded and Lipschitz.
\end{lemma}
\begin{proof}
    Since $f\in C^2\big(\R;(0,1)\big)$ with
    \[
        \lim_{x\to-\infty} f(x)=0\qquad\text{and}
        \qquad
        \lim_{x\to\infty} f(x)=1,
    \]
    we have
    \begin{equation}\label{eq.proof.lemma.A1}
    \int_{-\infty}^{\infty} f'(x)\,dx = 1.
    \end{equation}
    Moreover, since
    \[
        \dv{}{x}\ln f'(x)=g(x),
    \]
    and $g$ is bounded, there exists $M>0$ such that
    \[
        |\ln f'(x)-\ln f'(y)|\le M|x-y|
    \]
    for all $x,y\in\R$. Hence, we obtain
    \[
        f'(y)\ge f'(x)e^{-M|x-y|}
    \]
    for all $x,y\in\R$.\ Now suppose, towards a contradiction, that $f'$ is unbounded. Then there exists a sequence $(x_n)_{n\in\N}\subseteq\R$ such that $f'(x_n)\to\infty$.\ Since $f'(x)>0$ for all $x\in\R$, by \eqref{eq.proof.lemma.A1}, we obtain the contradiction
    \[
        1 \geq  \int_{x_n}^{\infty} f'(y)\, {\rm d}y \ge \int_{x_n}^{\infty} f'(x_n)e^{M(x_n-y)}\, {\rm d}y=\frac{f'(x_n)}{M}\to\infty.
    \]
    We have therefore shown that $f'$ is bounded.\ Hence, $f''=gf'$ is bounded as a product of bounded functions.\ If $g$ is, in addition, Lipschitz, it follows that $f''$ is bounded and Lipschitz as a product of bounded and Lipschitz functions.
\end{proof}

The following lemma reports a standard a priori estimate for the latent belief process, see, e.g.\ \cite[Proof of Theorem 5.2.1, p.\ 71]{oksendal}.\ For the sake of a self-contained exposition, we provide a short proof.

\begin{lemma}
Assume that Assumption \ref{ass.global} is satisfied.\ Then, there exists a constant $C_L\geq 1$ such that
\begin{equation}\label{eq.aprioi.spacetime}
\E\big[\big|L_s^{t_1,x_1}-L_s^{t_2,x_2}\big|\big]\leq C_L\big(|t_1-t_2|^{1/2}+|x_1-x_2|\big)
\end{equation}
for all $t_1,t_2,s\in [0,T]$ with $s\geq t_1\vee t_2$ and $x_1,x_2\in \R$.
\end{lemma}

\begin{proof}
We start with the case $t_1=t_2=:t$.\ Then, using the Lipschitz continuity of $\sigma$ and $g$,
\begin{align*}
\E\big[\big|L_s^{t,x_1}-L_s^{t,x_2}\big|^2\big]&\leq 3|x_1-x_2|^2+3\int_t^s \E\big[\big|\sigma(u,L_u^{t,x_1})-\sigma(u,L_u^{t,x_2})\big|^2\big]\, du\\
&\quad+3(s-t)\int_t^s \E\big[\big|\mu(u,L_u^{t,x_1})-\mu(u,L_u^{t,x_2})\big|^2\big]\, du\\
&\leq 3|x_1-x_2|^2+3C^2(1+T) \int_t^s \E\big[\big|L_u^{t,x_1}-L_u^{t,x_2}\big|^2\big]\, du.
\end{align*}
Using Gronwall's lemma, it follows that
\[
\E\big[\big|L_s^{t,x_1}-L_s^{t,x_2}\big|\big]\leq \sqrt{3}e^{\frac32C^2(1+T)T}|x_1-x_2|.
\]
Next, we prove the statement for $t_1\leq t_2$ and $x_1=x_2=:x$.\ Using the strong uniqueness of the SDE \eqref{eq:sde.latent} as well as the uniform boundedness of $\sigma$ and $g$,
\begin{align*}
\E\big[\big|L_s^{t_1,x}-L_s^{t_2,x}\big|\big]&\leq  \sqrt{3}e^{\frac32C^2(1+T)T} \E\big[\big|L_{t_2}^{t_1,x}-x\big|\big]\\
&\leq  \sqrt{3}e^{\frac32C^2(1+T)T}\big( C(t_2-t_1)^{1/2}+C^2\|g\|_\infty(t_2-t_1)\big)\\
&\leq \sqrt{3}e^{\frac32C^2(1+T)T}C\big(1+C\|g\|_\infty\sqrt{T}\big)|t_1-t_2|^{1/2}.
\end{align*}
Now, the claim follows by choosing
\[
C_L:=\sqrt{3}e^{\frac32C^2(1+T)T}\max\Big\{1,C\big(1+C\|g\|_\infty\sqrt{T}\big)\Big\}.
\]
\end{proof}

\section{Proof of Proposition \ref{lem:H_unique_maximizer_pred_model_general}}\label{sec:proof.unique.maximizer}

We next prove Proposition~\ref{lem:H_unique_maximizer_pred_model_general}, which characterizes the maximizers appearing in the Hamiltonians $H^b$ and $H^a$.\ In the proof, we adopt ideas from \cite{Gueant2017} and transfer them to our setup.

\begin{proof}[Proof of Proposition~\ref{lem:H_unique_maximizer_pred_model_general}]
Let $(t,p)\in[0,T]\times(0,1)$ and $\circ\in\{b,a\}$. For ease of notation, we write $\Lambda^\circ(\,\cdot\,):=\Lambda^\circ(t,p,\,\cdot\,)$, $G^\circ(\,\cdot\,):=G^\circ(p,\,\cdot\,)$, and $H^\circ(\,\cdot\,):=H^\circ(t,p;\,\cdot\,)$.\ Hence, for instance, $(\Lambda^\circ)'(\pi) = \partial_\pi \Lambda(t,p,\pi)$. In addition, we omit the subscript $_{t,p}$, so that, for example, $\Psi^\circ(\pi;z):=\Psi^\circ_{t,p}(\pi;z)$ and $u^\circ := u^\circ_{t,p}$.

\begin{enumerate}[(i)]
    \item For any $z\in\R$, the function $\pi\mapsto\Psi^\circ(\pi;z)$ is continuous on the compact interval $[0,1]$, so that it admits a maximizer.

    For $\pi\in(0,1)$, we have
    \[
        \partial_\pi \Psi^\circ(\pi;z)
        =(\Lambda^\circ)'(\pi)\,(G^\circ(\pi)-z)+\Lambda^\circ(\pi)\,(G^\circ)'(\pi).
    \]
    If $\circ=b$, then $(G^\circ)'(\pi)=-1$, so
    \begin{align*}
    \partial_\pi \Psi^b(\pi;z)
    &=(\Lambda^b)'(\pi)\bigg(G^b(\pi)-z-\frac{\Lambda^b(\pi)}{(\Lambda^b)'(\pi)}\bigg)\\
    &=(\Lambda^b)'(\pi)\bigl(u^b(\pi)-z\bigr).
    \end{align*}
    If $\circ=a$, then $(G^\circ)'(\pi)=1$, so
    \begin{align*}
    \partial_\pi \Psi^a(\pi;z)
    &=(\Lambda^a)'(\pi)\bigg(G^a(\pi)-z+\frac{\Lambda^a(\pi)}{(\Lambda^a)'(\pi)}\bigg)\\
    &=(\Lambda^a)'(\pi)\bigl(u^a(\pi)-z\bigr).
    \end{align*}
    Differentiating $u^\circ$ yields
    \[
        (u^b)'(\pi)
        =-2+\frac{\Lambda^b(\pi)\,(\Lambda^b)''(\pi)}
        {[(\Lambda^b)'(\pi)]^2}
    \]
    and
    \[
        (u^a)'(\pi)
        =2-\frac{\Lambda^a(\pi)\,(\Lambda^a)''(\pi)}
        {[(\Lambda^a)'(\pi)]^2}.
    \]
    Hence, by the curvature condition, we obtain $(u^b)'(\pi)<-\kappa<0$ and $(u^a)'(\pi)>\kappa>0$ for every $\pi\in(0,1)$ and a constant $\kappa>0$ independent of $\pi$. We use these monotonicity properties in the following cases.
    
    \emph{Case 1}:\ Let $z\in u^\circ\big((0,1)\big)$.\ Since $(\Lambda^b)'(\pi)$ is strictly positive and $(\Lambda^a)'(\pi)$ is strictly negative, we find that $\pi\mapsto\Psi^\circ(\pi;z)$ is strictly increasing up to the point where $u^\circ(\pi)=z$ and strictly decreasing afterwards if such a $\pi\in(0,1)$ exists. Indeed, since $z\in u^\circ((0,1))$ and $u^\circ$ is strictly monotone on $(0,1)$, there exists a unique $\pi\in(0,1)$ such that $u^\circ(\pi)=z$, which is given by $\pi^{\circ,*}(z) = \l(u^\circ\r)^{-1}(z)$.
    
    \emph{Case 2}:\ Let $z\notin u^\circ\big((0,1)\big)$.\ Consider the case $\circ=b$. Since $u^b$ is strictly decreasing on $(0,1)$, we have:
    \begin{itemize}
    \item If $z\ge u^b(0+)$: Then $u^b(\pi)-z<0$ for all $\pi\in(0,1)$ and hence
    $\partial_\pi\Psi^b(\pi;z)<0$ on $(0,1)$. Since $\Psi^b(\,\cdot\,;z)$ is continuous on $[0,1]$, it follows that it is strictly decreasing on $[0,1]$. Thus, the unique maximizer is given by $\pi^{b,*}(z)=0$.
    \item If $z\le u^b(1-)$: Then $u^b(\pi)-z>0$ for all
    $\pi\in(0,1)$ and hence
    $\partial_\pi\Psi^b(\pi;z)>0$ on $(0,1)$. Similar to above, we find that the unique maximizer is given by $\pi^{b,*}(z)=1$.
    \end{itemize}

    The argument for $\circ=a$ is analogous.

    \item Set $J^\circ:=u^\circ\big((0,1)\big)$. From part (i) we know that for any $z\in J^\circ$, we have $\pi^{\circ,*}(z)\in(0,1)$ and $u^\circ\big(\pi^{\circ,*}(z)\big) = z$.\
    Define
    $$F(z,\pi) := u^\circ(\pi)-z\qquad\text{for }(z,\pi)\in J^\circ\times(0,1).$$
    Then, $F\l(z,\pi^{\circ,*}(z)\r) = 0$ and $\partial_\pi F(z,\pi) = \l(u^\circ\r)'(\pi)\neq0$.\ Since $\Lambda^\circ$ is of class $C^2$, $\l(u^\circ\r)'$ is continuous and hence $F$ is continuously differentiable. By the implicit function theorem, the map $z\mapsto\pi^{\circ,*}(z)$ is of class $C^1$ on $J^\circ$ and
    \[
        \l(\pi^{\circ,*}\r)'(z) = \frac{1}{\l(u^\circ\r)'\big(\pi^{\circ,*}(z)\big)}.
    \]
    Since $\l(u^b\r)'<0$ and $\l(u^a\r)'>0$, the claim follows.

    \item Let $z_1,z_2\in\R$ with $z_1<z_2$ and $\pi\in[0,1]$. Then,
    \[
    \Psi(\pi;z_1)
    =\Lambda^\circ(\pi)\bigl(G^\circ(\pi)-z_1\bigr)
    \ge \Lambda^\circ(\pi)\bigl(G^\circ(\pi)-z_2\bigr)
    =\Psi^\circ(\pi;z_2).
    \]
    Taking the supremum over $\pi\in[0,1]$ and multiplying by $\Delta>0$ yields
    $H^\circ(z_1)\ge H^\circ(z_2)$. Now let $z\in J^\circ$. Define
    \[
        K^\circ(z,\pi):= \Delta \Lambda^\circ(\pi)(G^\circ(\pi)-z) \qquad \text{for }(z,\pi)\in J^\circ\times[0,1].
    \]
    Then, the envelope theorem gives
    \[
        (H^\circ)'(z)
        = \partial_z K^\circ\l(z,\pi^{\circ,*}(z)\r)
        = -\Delta \Lambda^\circ\l(\pi^{\circ,*}(z)\r),
    \]
    and, since $\pi^{\circ,*}$ is of class $C^1$ on $J^\circ$, we obtain that $H^\circ$ is of class $C^2$ on $J^\circ$.
    \item By (iii), we have
    \[
        -(H^\circ)'(z)
        =\Delta \Lambda^\circ\big(\pi^{\circ,*}(z)\big)
    \]
    for $z\in J^\circ$. Since $\Lambda^\circ$ is strictly monotone on $(0,1)$, it is invertible on its image. Thus, we obtain
    \[
        \pi^{\circ,*}(z) = \l(\Lambda^\circ\r)^{-1}\l(-\frac{(H^\circ)'(z)}{\Delta}\r).
    \]
    \item Let $(t,p,z)\in[0,T]\times(0,1)\times\R$ and let $(t_n,p_n,z_n)_{n\in\N}\subseteq[0,T]\times(0,1)\times\R$ with $(t_n,p_n,z_n)\longrightarrow(t,p,z)$. Set
    \[
        \pi_n:=\pi_{t_n,p_n}^{\circ,*}(z_n).
    \]
    Let $(\pi_{n_k})_{k\in\N}$ be an arbitrary convergent subsequence of
    $(\pi_n)_{n\in\N}$, with
    \[
        \pi_{n_k}\to\bar\pi\in[0,1].
    \]
    By the optimality of $\pi_{n_k}$, for every $\pi\in[0,1]$,
    \[
        \Psi_{t_{n_k},p_{n_k}}^{\circ}(\pi_{n_k};z_{n_k})
        \geq
        \Psi_{t_{n_k},p_{n_k}}^{\circ}(\pi;z_{n_k}).
    \]
    Since the map
    \[
        (t,p,z,\pi)\to\Psi_{t,p}^{\circ}(\pi;z)
    \]
    is continuous, passing to the limit yields
    \[
        \Psi_{t,p}^{\circ}(\bar\pi;z)
        \ge
        \Psi_{t,p}^{\circ}(\pi;z)
    \]
    for every $\pi\in[0,1]$. Hence, $\bar\pi$ is a maximizer of
    $\pi\mapsto\Psi_{t,p}^{\circ}(\pi;z)$. By uniqueness of the maximizer
    established in part~(i),
    \[
        \bar\pi=\pi_{t,p}^{\circ,*}(z).
    \]
    Thus, every convergent subsequence of $(\pi_n)_{n\in\N}$ has limit
    $\pi_{t,p}^{\circ,*}(z)$. Since $(\pi_n)_{n\in\N}$ takes values in the
    compact set $[0,1]$, it follows that
    \[
        \pi_n\to\pi_{t,p}^{\circ,*}(z)\qquad \text{as }n\to \infty.
    \]
    Consequently, the map $(t,p,z)\mapsto\pi_{t,p}^{\circ,*}(z)$
    is continuous.
\end{enumerate}
\end{proof}

\section{Proof of Theorem \ref{thm:verification}}\label{app:proof.verification}

\begin{proof}[Proof of Theorem \ref{thm:verification}]
Fix $(t,p,q,x)\in[0,T]\times(0,1)\times\mathcal Q\times\R$. If
$t=T$, the result follows immediately from the terminal condition.\ Hence we assume that $t<T$. Let $\pi\in\mathcal A(t)$ be arbitrary.\ For ease of notation, we write
\[
        X_s=X_s^{t,p,q,x,\pi},\qquad
        q_s=q_s^{t,p,q,\pi},\qquad\text{and}\qquad
        p_s=p_s^{t,p}.
\]

Define the jump increments
\[
J_s^b := \Upsilon(s,p_s,q_{s-}+\Delta,X_{s-}-\Delta\pi_s^b)
-
\Upsilon(s,p_s,q_{s-},X_{s-})
\]
and
\[
J_s^a :=
\Upsilon(s,p_s,q_{s-}-\Delta,X_{s-}+\Delta\pi_s^a)
-
\Upsilon(s,p_s,q_{s-},X_{s-}).
\]
The processes $J^b$ and $J^a$ are predictable.\ Moreover, since $p_s\in(0,1)$, $\pi_s^a,\pi_s^b\in[0,1]$, and $V$ is bounded, there
is a constant $C_J<\infty$ such that
\[
        |J_s^b|+|J_s^a|\le C_J.
\]
Together with the boundedness of $\Lambda^b$ and $\Lambda^a$, this implies
\[
        \E\bigg[
        \int_t^T |J_s^b|\lambda_s^b\,ds
        + \int_t^T |J_s^a|\lambda_s^a\,ds
        \bigg]<\infty,
\]
whereby $\lambda_s^b = \Lambda^b(s,p_s,\pi_s^b)\1_{\{q_{s-}<Q\}}$ and $\lambda_s^a = \Lambda^a(s,p_s,\pi_s^a)\1_{\{q_{s-}>-Q\}}$.\ Therefore, the compensated jump integrals
\[
        \int_t^\cdot J_s^b\,\big(dN_s^b-\lambda_s^b\,ds\big) \qquad \text{and}\qquad
        \int_t^\cdot J_s^a\,\big(dN_s^a-\lambda_s^a\,ds\big)
\]
are martingales, so that
\[
    \E\bigg[\int_t^T J_s^b\, dN_s^b + \int_t^T J_s^a \,dN_s^a\bigg] = \E\bigg[\int_t^T J_s^b \lambda_s^b\,ds + \int_t^T J_s^a \lambda_s^a\,ds\bigg].
\]

To localize the Brownian part, let $\ell=f^{-1}(p)$, set
\[
        R_n:=n+|x|, \qquad \theta_n:=T-\frac{T-t}{n+1},
\]
and define 
\[
    \tau_n := \inf\big\{s\in[t,T]: |L^{t,\ell}_s|\ge R_n\big\}\wedge \theta_n.
\]
Since $L^{t,\ell}$ is continuous on $[t,T]$, we have $\tau_n\uparrow T$ almost
surely.

Set
\[
    \eta_s := \big(q_{s-}+\partial_pV(s,p_s,q_{s-})\big)\varsigma(s,p_s).
\]
On $[t,\tau_n]$, we have $s\le\theta_n<T$ and
$p_s=f(L_s^{t,\ell})\in f\big([-R_n,R_n]\big)$, which is a compact subset of $(0,1)$. Since $V$ is a classical solution to \eqref{eq:HJB_PDE_reduced_pred}, $\partial_pV$ is continuous and therefore bounded on the compact set
\[
        [t,\theta_n]\times f\big([-R_n,R_n]\big)\times\mathcal Q.
\]
Thus $\eta$ is bounded on $[t,\tau_n]$. Consequently,
\[
    \E\bigg[\int_t^{\tau_n}\eta_s\,dW_s\bigg]=0.
\]
Applying Itô's formula to $\Upsilon$ on the stopped interval $[t,\tau_n]$ gives
\[
\begin{aligned}
&\Upsilon(\tau_n,p_{\tau_n},q_{\tau_n-},X_{\tau_n-})
=
\Upsilon(t,p,q,x)\\
&\qquad+
\int_t^{\tau_n}
\left(
\partial_t V(s,p_s,q_{s-})
+ \frac{1}{2}\varsigma(s,p_s)^2
\partial_{pp}^2V(s,p_s,q_{s-})
\right)ds
\\
&\qquad
+ \int_t^{\tau_n}
\big(q_{s-}+\partial_pV(s,p_s,q_{s-})\big)
\varsigma(s,p_s)\,dW_s
\\
&\qquad
+
\int_t^{\tau_n}
\Big[\Upsilon(s,p_s,q_{s-}+\Delta,X_{s-}-\Delta\pi_s^b)
- \Upsilon(s,p_s,q_{s-},X_{s-})\Big]dN_s^b
\\
&\qquad
+
\int_t^{\tau_n}
\Big[\Upsilon(s,p_s,q_{s-}-\Delta,X_{s-}+\Delta\pi_s^a)
- \Upsilon(s,p_s,q_{s-},X_{s-}) \Big]dN_s^a,
\end{aligned}
\]
so that taking expectations yields
\begin{align}\label{eq:upsilon_eq_verification_theorem}
    \E\big[\Upsilon(\tau_n,p_{\tau_n},q_{\tau_n-},X_{\tau_n-})\big] = \Upsilon(t,p,q,x) + \E\bigg[\int_t^{\tau_n} \mathcal{D}_s^\pi ds\bigg],
\end{align}
where
\[
    \mathcal{D}_s^\pi = \partial_t V(s,p_s,q_{s-})
    + \frac{1}{2}\varsigma(s,p_s)^2 \partial_{pp}^2V(s,p_s,q_{s-})
    + \lambda_s^bJ_s^b + \lambda_s^aJ_s^a.
\]
Whenever $q_{s-}<Q$, we have
\[
\begin{aligned}
J_s^b
&=
-\Delta\pi_s^b+\Delta p_s
+ V(s,p_s,q_{s-}+\Delta)-V(s,p_s,q_{s-})
\\
&=
\Delta
\left(
p_s-\pi_s^b
- \frac{V(s,p_s,q_{s-})-V(s,p_s,q_{s-}+\Delta)}{\Delta}
\right)
\\
&=
\Delta
\big(
G^b(p_s,\pi_s^b)-z_b(s,p_s,q_{s-})
\big).
\end{aligned}
\]
Similarly, whenever $q_{s-}>-Q$,
\[
    J_s^a = \Delta\big(G^a(p_s,\pi_s^a)-z_a(s,p_s,q_{s-})\big).
\]
Hence, we obtain
\[
\begin{aligned}
\mathcal{D}_s^\pi
&=
\partial_t V(s,p_s,q_{s-})
+ \frac{1}{2}\varsigma(s,p_s)^2
\partial_{pp}^2V(s,p_s,q_{s-})
\\
&\quad
+ \1_{\{q_{s-}<Q\}}
\Delta\Lambda^b(s,p_s,\pi_s^b)
\bigl(G^b(p_s,\pi_s^b)-z_b(s,p_s,q_{s-})\bigr)
\\
&\quad
+ \1_{\{q_{s-}>-Q\}}
\Delta\Lambda^a(s,p_s,\pi_s^a)
\bigl(G^a(p_s,\pi_s^a)-z_a(s,p_s,q_{s-})\bigr)\\
&\le \partial_t V(s,p_s,q_{s-})
+ \frac{1}{2}\varsigma(s,p_s)^2
\partial_{pp}^2V(s,p_s,q_{s-})
\\
&\quad
+ \1_{\{q_{s-}<Q\}}
H^b\big(s,p_s;z_b(s,p_s,q_{s-})\big)\\
&\quad
+ \1_{\{q_{s-}>-Q\}}
H^a\big(s,p_s;z_a(s,p_s,q_{s-})\big).
\end{aligned}
\]
Using the reduced Hamilton--Jacobi--Bellman equation \eqref{eq:HJB_PDE_reduced_pred}, this yields
\[
        \mathcal{D}_s^\pi
        \le
        \gamma q_{s-}^2\varsigma(s,p_s)^2.
\]
Combining this with \eqref{eq:upsilon_eq_verification_theorem}, using that $q_s=q_{s-}$ and $X_s=X_{s-}$ for Lebesgue-a.e.\ $s\in (t,T)$, and observing that
\[
\Delta N^b_{\tau_n}=\Delta N^a_{\tau_n}=0
\qquad \P\text{-a.s.},
\]
we obtain
\[
\begin{aligned}
\E\left[
\Upsilon(\tau_n,p_{\tau_n},q_{\tau_n},X_{\tau_n}) - \gamma\int_t^{\tau_n}q_s^2\varsigma(s,p_s)^2\,ds
\right]
\le
\Upsilon(t,p,q,x).
\end{aligned}
\] 
By dominated convergence, letting $n\to\infty$ yields
\[
\begin{aligned}
\E\bigg[
\Upsilon(T,p_T,q_T,X_T) - \gamma\int_t^Tq_s^2\varsigma(s,p_s)^2\,ds
\bigg] \le \Upsilon(t,p,q,x),
\end{aligned}
\]
where we used the fact that, by \eqref{eq:value_reduction},
\[
    \Upsilon(t,p,q,x) = x + qp + V(t,p,q),
\]
$V$ is bounded, and $\E[X_{\tau_n}]\to \E[X_T]$ as $n\to \infty$ since $\Lambda^\circ$ is uniformly bounded for $\circ\in\{b,a\}$ and jumps of $X$ are of size at most $\Delta$ $\P$-a.s.\ Using the terminal condition
\[
        V(T,p,q)=\Phi(p,q),
\]
we obtain
\[
\E\bigg[X_T+q_Tp_T +\Phi(p_T,q_T)- \gamma\int_t^Tq_s^2\varsigma(s,p_s)^2\,ds\bigg] \le \Upsilon(t,p,q,x).
\]
Since \(\pi\in\mathcal A(t)\) was arbitrary, this proves
\[
        \sup_{\pi\in\mathcal A(t)}\E\bigg[X_T+q_Tp_T +\Phi(p_T,q_T)- \gamma\int_t^Tq_s^2\varsigma(s,p_s)^2\,ds\bigg] \le \Upsilon(t,p,q,x).
\]

It remains to prove equality.\ Since $q_{s-}$ is predictable, $p$ is continuous and adapted, and the unique, up to the boundary values, optimizer $\pi^{\circ,*}$ is continuous by Proposition~\ref{lem:H_unique_maximizer_pred_model_general}, and hence measurable, the process
\[
    \pi_s^* = \big(\pi^{b,*}(s,p_s,q_{s-}), \pi^{a,*}(s,p_s,q_{s-})\big)
\]
is predictable and takes values in $[0,1]^2$.\ Hence, $\pi^*\in\mathcal A(t)$ and, by Proposition \ref{lem:H_unique_maximizer_pred_model_general}, for this control, all inequalities above are equalities.\ Therefore, the upper bound is attained, which completes the proof.
\end{proof}

\section{Proof of Theorem \ref{thm.ex.unique.classical}}\label{app:deferred_existence_and_uniqueness_of_a_mild_solution}

\subsection{A priori estimate for the Hamiltonian}

The following lemma provides an a priori estimate for the operator $\mathcal H$, which will be used to establish existence and uniqueness of a mild solution to \eqref{eq:HJB_PDE_reduced_pred}.\ In the sequel, let
\[
\varsigma_f(t,x):=f'(x)\sigma(t,x)\quad \text{for }t\in [0,T]\text{ and }x\in \R.
\]
Since $f$ is bounded and $f'$ is bounded, we have $[f]_\alpha<\infty$ and $[\varsigma_f^2]_{\alpha/2,\alpha}<\infty$.

\begin{lemma}\label{lem:operator.calH}
Assume that Assumption \ref{ass.exunique} is satisfied.\ Let $\bar{\Lambda}>0$ be the common bound for $\Lambda^b$ and $\Lambda^a$, and $M\ge0$.\ Then, for all $t_1,t_2\in[0,T]$, $x_1,x_2\in\R$, 
$q\in Q$, and all bounded functions
$u_1,u_2\colon [0,T]\times\R\times Q\to\R$ with
\[
\max\big\{\|u_1(t_1,\,\cdot\,,\,\cdot\,)\|_\infty,\|u_2(t_2,\,\cdot\,,\,\cdot\,)\|_\infty\big\}
\le M,
\]
we have
\[
\begin{aligned}
&\big|\mathcal H u_1(t_1,x_1,q)-\mathcal H u_2(t_2,x_2,q)\big| \\
&\le
C_{\mathcal H}(M)
\big(|t_1-t_2|^{\alpha/2}+|x_1-x_2|^\alpha\big)
+
4\bar\Lambda
\sup_{\bar q\in Q}
|u_1(t_1,x_1,\bar q)-u_2(t_2,x_2,\bar q)|,
\end{aligned}
\]
where
\[
C_{\mathcal H}(M)
:=
\gamma Q^2[\varsigma_f^2]_{\alpha/2,\alpha}
+
\Delta\sum_{\circ\in\{b,a\}}
\bigg(
\big[\Lambda_f^\circ\big]_{\alpha/2,\alpha}
\bigg(1+\frac{2M}{\Delta}\bigg)
+
\bar\Lambda [f]_\alpha
\bigg).
\]
\end{lemma}

\begin{proof}
We first prove an auxiliary estimate for the Hamiltonians.\ To that end, let
$\circ\in\{b,a\}$, $R\ge0$, $t_1,t_2\in[0,T]$, $x_1,x_2\in\R$, and $z\in[-R,R]$. We obtain
\[
\begin{aligned}
\big|H^\circ\big(t_1,f(x_1);z\big)-H^\circ&\big(t_2,f(x_2);z\big)\big| \\
&\le
\Delta
\sup_{\pi\in[0,1]}
\Big|
\Lambda^\circ\big(t_1,f(x_1),\pi\big)
\Big(G^\circ\big(f(x_1),\pi\big)-z\Big)\\
&\quad-
\Lambda^\circ\big(t_2,f(x_2),\pi\big)
\Big(G^\circ\big(f(x_2),\pi\big)-z\Big)
\Big|.
\end{aligned}
\]
For fixed $\pi\in[0,1]$, we have
\[
\begin{aligned}
\Big|
\Lambda^\circ\big(t_1,f(x_1),\pi\big)
\Big(&G^\circ\big(f(x_1),\pi\big)-z\Big)
-
\Lambda^\circ\big(t_2,f(x_2),\pi\big)
\Big(G^\circ\big(f(x_2),\pi\big)-z\Big)
\Big| \\
&\le
\big|\Lambda^\circ\big(t_1,f(x_1),\pi\big)
-
\Lambda^\circ\big(t_2,f(x_2),\pi\big)\big|
\,\big|G^\circ\big(f(x_1),\pi\big)-z\big| \\
&\quad+
\big|\Lambda^\circ\big(t_2,f(x_2),\pi\big)\big|
\,\big|G^\circ\big(f(x_1),\pi\big)-G^\circ\big(f(x_2),\pi\big)\big|.
\end{aligned}
\]
Since, for any $\pi\in[0,1]$, $\big|G^\circ\big(f(\,\cdot\,),\pi\big)\big|\le1$,
$|z|\le R$, and
\[
\big|G^\circ\big(f(x_1),\pi\big)-G^\circ\big(f(x_2),\pi\big)\big|
=
|f(x)-f(y)|
\le
[f]_\alpha |x-y|^\alpha,
\]
it follows that
\begin{equation}\label{eq:H_estimate_proof_eq}
\begin{aligned}
&\big|H^\circ\big(t_1,f(x_1);z\big)-H^\circ\big(t_2,f(x_2);z\big)\big| \\
&\le
\Delta
\left(
\big[\Lambda_f^\circ\big]_{\alpha/2,\alpha}(1+R)
+
\bar\Lambda [f]_\alpha
\right)
\big(|t_1-t_2|^{\alpha/2}+|x_1-x_2|^\alpha\big).
\end{aligned}
\end{equation}

Now let $u_1,u_2$ be bounded and assume that
\[
\max\big\{\|u_1(t_1,\,\cdot\,,\,\cdot\,)\|_\infty,\|u_2(t_2,\,\cdot\,,\,\cdot\,)\|_\infty\big\}
\le M.
\]
Whenever the neighboring inventory levels are admissible, define
\begin{align*}
&D_q^+u(t_1,x_1)
:=
\frac{u(t_1,x_1,q)-u(t_1,x_1,q+\Delta)}{\Delta},\\
&D_q^-u(t_1,x_1)
:=
\frac{u(t_1,x_1,q)-u(t_1,x_1,q-\Delta)}{\Delta},
\end{align*}
and analogously $D_q^+u_2(t_2,x_2)$ and $D_q^-u_2(t_2,x_2)$.\ Since
$\|u_2(t_2,\,\cdot\,,\,\cdot\,)\|_\infty\le M$, we have
\[
|D_q^+u_2(t_2,x_2)|\le \frac{2M}{\Delta},
\qquad
|D_q^-u_2(t_2,x_2)|\le \frac{2M}{\Delta}.
\]
We estimate the bid term.\ Notice that
\[
\begin{aligned}
&\big|H^b\big(t_1,f(x_1);D_q^+u_1(t_1,x_1)\big)
-
H^b\big(t_2,f(x_2);D_q^+u_2(t_2,x_2)\big)\big| \\
&\le
\big|H^b\big(t_1,f(x_1);D_q^+u_1(t_1,x_1)\big)
-
H^b\big(t_1,f(x_1);D_q^+u_2(t_2,x_2)\big)\big| \\
&\quad+
\big|H^b\big(t_1,f(x_1);D_q^+u_2(t_2,x_2)\big)
-
H^b\big(t_2,f(x_2);D_q^+u_2(t_2,x_2)\big)\big|.
\end{aligned}
\]
By definition of $H^b$,
\[
\begin{aligned}
&\big|H^b(t_1,f(x_1);D_q^+u_1(t_1,x_1)\big)
-
H^b\big(t_1,f(x_1);D_q^+u_2(t_2,x_2)\big)\big| \\
&\le
\Delta\bar\Lambda
|D_q^+u_1(t_1,x_1)-D_q^+u_2(t_2,x_2)| \\
&\le
2\bar\Lambda
\sup_{\bar q\in Q}
|u_1(t_1,x_1,\bar q)-u_2(t_2,x_2,\bar q)|.
\end{aligned}
\]
Moreover, by \eqref{eq:H_estimate_proof_eq} with $R=\frac{2M}{\Delta}$,
\[
\begin{aligned}
&|H^b(t_1,f(x_1);D_q^+u_2(t_2,x_2))
-
H^b(t_2,f(x_2);D_q^+u_2(t_2,x_2))| \\
&\le
\Delta
\bigg(
\big[\Lambda_f^b\big]_{\alpha/2,\alpha}
\bigg(1+\frac{2M}{\Delta}\bigg)
+
\bar\Lambda [f]_\alpha
\bigg)
\big(|t-s|^{\alpha/2}+|x-y|^\alpha\big).
\end{aligned}
\]
A similar argument applies to the ask term.\ Finally, the running term satisfies
\[
\begin{aligned}
&\gamma q^2
\big|\varsigma_f(t_1,x_1)^2-\varsigma_f(t_2,x_2)^2\big| \le
\gamma q^2[\varsigma_f^2]_{\alpha/2,\alpha}
\big(|t_1-t_2|^{\alpha/2}+|x_1-x_2|^\alpha\big).
\end{aligned}
\]
Combining the running, bid, and ask estimates gives
\[
\begin{aligned}
\big|\mathcal H u_1(t_1,x_1,q)-\mathcal H u_2(t_2,x_2,q)\big| &\le
C_{\mathcal H}(M)
\big(|t_1-t_2|^{\alpha/2}+|x_1-x_2|^\alpha\big)\\
&\quad +
4\bar\Lambda
\sup_{\bar q\in Q}
|u_1(t_1,x_1,\bar q)-u_2(t_2,x_2,\bar q)|.
\end{aligned}
\]
The proof is complete.
\end{proof}

\begin{remark}
The preceding lemma yields two useful special cases.

First, choosing $t_1=t_2=:t$ and $x_1=x_2=:x$, we obtain
\[
\big|\mathcal H u_1(t,x,q)-\mathcal H u_2(t,x,q)\big|
\le
4\bar\Lambda
\sup_{\bar q\in Q}
|u_1(t,x,\bar q)-u_2(t,x,\bar q)|.
\]

Second, choosing $u_1=u_2=:u$, we get
\[
\begin{aligned}
&\big|\mathcal H u(t_1,x_1,q)-\mathcal H u(t_2,x_2,q)\big| \\
&\le
C_{\mathcal H}(M)
\big(|t_1-t_2|^{\alpha/2}+|x_1-x_2|^\alpha\big) \\
&\quad+
4\bar\Lambda
\sup_{\bar q\in Q}
|u(t_1,x_1,\bar q)-u(t_2,x_2,\bar q)|.
\end{aligned}
\]
In particular, if $u$ is parabolic $\alpha$-Hölder continuous in $(t,x)$, then $\mathcal H u$ is parabolic $\alpha$-Hölder continuous in $(t,x)$.
\end{remark}

\subsection{Existence and Uniqueness of a Mild Solution}\label{app:deferred_existence_and_uniqueness_of_a_mild_solution.mild}
We now show the existence and uniqueness of a mild solution of the form \eqref{eq:def.mild.solution} to the transformed version of the HJB equation \eqref{eq:HJB_PDE_reduced_pred} with terminal condition \eqref{eq:terminal_pred_reduced} using Banach's fixed point theorem.

To that end, for $u\in \mathcal U:=C_b\big([0,T]\times \R\times \mathcal Q\big)$ and $(t,x,q)\in[0,T]\times\R\times\mathcal Q$, let
\[
 (\Gamma u)(t,x,q):=\E\big[\varphi\big(L_T^{t,x},q\big)\big]-
    \int_t^T
    \E\big[(\mathcal Hu)\big(s,L_s^{t,x},q\big)\big]\,ds,
\]
where
\[
\varphi(x,q):=\Phi\big(f(x),q\big)
\]
for $x\in \R$, $q\in \mathcal Q$, and fixed $\Phi\in C^\alpha\big((0,1)\times \mathcal Q\big)$.\ 

To prove the existence and uniqueness of a fixed point of $\Gamma$, we equip $\mathcal U$ with a weighted supremum norm.\ For $\beta>0$, define
\[
    \|u\|_\beta
    := \sup_{(t,x,q)\in [0,T]\times\R\times\mathcal Q}e^{-\beta(T-t)}|u(t,x,q)|.
\]
This norm is equivalent to the supremum norm since
\[
    e^{-\beta T}\|u\|_\infty \le \|u\|_{\beta} \le \|u\|_\infty\quad \text{for all }u\in \mathcal U.
\]
 In particular, $(\mathcal U,\|\cdot\|_\beta)$ is a Banach space. The exponential weight allows us to show that $\Gamma$ is a contraction with respect to the weighted supremum norm for sufficiently large $\beta$.\ Moreover, for $\alpha\in (0,1)$ and $u\in \mathcal U$, we consider the weighted H\"older seminorm
\[
[u]_{\alpha,\beta}:= \sup_{(t,q)\in [0,T]\times \mathcal Q}\sup_{\substack{x_1,x_2\in \R\\ x_1\neq x_2}}e^{-\beta(T-t)}\frac{|u(t,x_1,q)-u(t,x_2,q)|}{|x_1-x_2|^\alpha}\in [0,\infty].
\]
To ease notation, we set $[u]_\alpha:=[u]_{\alpha,0}$ for $u\in \mathcal U$.\ Again, we have
\[
e^{-\beta T}[u]_{\alpha} \le [u]_{\alpha,\beta} \le [u]_{\alpha}\quad \text{for }u\in \mathcal U.
\]
We now prove the main result of this subsection.

\begin{theorem}[Existence and Uniqueness of a Mild Solution]
\label{thm:existence_uniqueness_mild}
Assume that Assumption \ref{ass.exunique} is satisfied.\ Then, the transformed version of the reduced Hamilton--Jacobi--Bellman equation \eqref{eq:HJB_PDE_reduced_pred} with terminal condition \eqref{eq:terminal_pred_reduced} admits a unique mild solution.\ Moreover, the unique mild solution $v$ satisfies $|v|_{\alpha}<\infty$.\ Equivalently, the operator $\Gamma$ has a unique fixed point, and the fixed point $v$ satisfies $|v|_{\alpha}<\infty$.
\end{theorem}

\begin{proof}
Let $u_1,u_2\in\mathcal U$ and $(t,x,q)\in[0,T]\times\R\times \mathcal Q$.\
By definition of $\Gamma$ and Lemma \ref{lem:operator.calH}, we have
\begin{align*}
|(\Gamma u_1)(t,x,q)-&(\Gamma u_2)(t,x,q)|\\
&\quad\le
\int_t^T \E\big[\big|(\mathcal Hu_1)\big(s,L_s^{t,x},q\big)-(\mathcal Hu_2)\big(s,L_s^{t,x},q\big)\big|\big]\,ds\\
&\quad\leq 4\overline\Lambda\int_t^T \E\bigg[\sup_{\bar q\in \mathcal Q}\big|u_1\big(s,L_s^{t,x},\bar q\big)-u_2\big(s,L_s^{t,x},\bar q\big)\big|\bigg]\,ds.
\end{align*}
Hence,
\[
    \big|(\Gamma u_1)(t,x,q)-(\Gamma u_2)(t,x,q)\big|
    \le
    4\overline{\Lambda} \|u_1-u_2\|_\beta\int_t^T e^{\beta(T-s)}\,ds.
\]
Multiplying both sides by $e^{-\beta(T-t)}$ and taking the supremum over $(t,x,q)\in [0,T]\times \R\times \mathcal Q$, yields
\begin{equation}\label{eq:bound.beta.norm}
\|\Gamma u_1-\Gamma u_2\|_\beta\leq \frac{4\overline{\Lambda}}{\beta}\|u_1-u_2\|_\beta.
\end{equation}
Now, let $\alpha\in (0,1)$.\ Repeating a similar estimate with Lemma \ref{lem:operator.calH}, for $u\in \mathcal U$ with $[u]_{\alpha}<\infty$, $(t,q)\in [0,T]\times \mathcal Q$, and $x_1,x_2\in \R$, we also have
\begin{align*}
 e^{-\beta(T-t)}|(\Gamma u)(t,x_1,q)-&(\Gamma u)(t,x_2,q)| \\
 &\le C_L^\alpha\bigg([\varphi]_{\alpha}+\frac{C_{\mathcal H}\big(\|u\|_\beta\big)}{\beta}+\frac{4\overline\Lambda}{\beta}[u]_{\alpha,\beta}\bigg)|x_1-x_2|^\alpha,
\end{align*}
 where $C_L\geq 1$ is the constant from \eqref{eq.aprioi.spacetime}.\ This implies that, for $u\in \mathcal U$ with $[u]_{\alpha}<\infty$,
 \begin{equation}\label{eq:estimate.hoelder.beta}
 \big[\Gamma u\big]_{\alpha,\beta}\leq C_L^\alpha\bigg([\varphi]_{\alpha}+\frac{C_{\mathcal H}\big(\|u\|_\beta\big)}{\beta}+\frac{4\overline{\Lambda}}{\beta}[u]_{\alpha,\beta}\bigg).
 \end{equation}

Now, we choose $\beta>0$ such that $C_L^\alpha 4\overline\Lambda= \frac\beta2$.\ Then, by the triangle inequality and \eqref{eq:bound.beta.norm} with $u_1=u\in \mathcal U$ and $u_2=0$,
\[
\|\Gamma u\|_{\beta}\leq \|\Gamma 0\|_\infty+\frac12\| u\|_{\beta}.
\]
Iteratively, starting from $u_0=0\in \mathcal U$ and defining $u_{n}:=\Gamma u_{n-1}$ for all $n\in \N$, we thus obtain
\[
\|u_n\|_{\beta}\leq \|\Gamma 0\|_{\infty}\sum_{k=0}^{n-1} 2^{-k}\leq 2\|\Gamma 0\|_\infty\quad \text{for all }n\in \N,
\]
and therefore, by \eqref{eq:estimate.hoelder.beta},
\[
[u_n]_{\alpha,\beta}\leq C_L^\alpha\bigg([\varphi]_{\alpha}+\frac{C_{\mathcal H}\big(2\|\Gamma 0\|_\infty\big)}{\beta}\bigg)\sum_{k=0}^{n-1} 2^{-k}\leq C_{\alpha,\beta}\quad \text{for all }n\in \N
\]
with
\[
C_{\alpha,\beta}:=2C_L^\alpha\bigg([\varphi]_{\alpha}+\frac{C_{\mathcal H}\big(2\|\Gamma 0\|_\infty\big)}{\beta}\bigg).
\]
By Banach's fixed point theorem, the sequence $(u_n)_{n\in \N}\subset \mathcal U$ converges to the unique fixed point $v\in \mathcal U$ of $\Gamma$ w.r.t.\ the supremum norm.\ Since $[u_n]_{\alpha,\beta}\leq C_{\alpha,\beta}$, it follows that $$[v]_{\alpha}\leq e^{\beta T}[v]_{\alpha,\beta}\leq e^{\beta T}C_{\alpha,\beta}.$$
By definition of $\Gamma$, the fixed point $v$ is the unique mild solution to the transformed version of \eqref{eq:HJB_PDE_reduced_pred} with terminal condition
\eqref{eq:terminal_pred_reduced}.\ The proof is complete.
\end{proof}

\subsection{Existence and Uniqueness of a Classical Solution}\label{app:proof.classical}

Combining the regularity results in \cite[Chapter 9]{Krylov1996} with Theorem \ref{thm:existence_uniqueness_mild}, we are now in a position to prove Theorem \ref{thm.ex.unique.classical}.

\begin{proof}[Proof of Theorem \ref{thm.ex.unique.classical}]
In a first step, we prove that every classical solution $$v\in C^{1,2}\big([0,T)\times \R\times \mathcal Q\big)\cap C_{\rm b}\big([0,T]\times \R\times \mathcal Q\big)$$ to the abstract Cauchy problem
\begin{equation}\label{eq:lin.classical}
\begin{aligned}
 \partial_t v(t,x,q)+\big(\mathcal L_t v(t,\, \cdot\, ,q)\big)(x)&=h(t,x,q),\\
 v(T,x,q)&=\varphi(x,q),
\end{aligned}
\end{equation}
for $(t,x,q)\in [0,T)\times \R\times \mathcal Q$, where $h=\mathcal Hu$ with $u\in C_{\rm b}\big([0,T]\times \R\times \mathcal Q\big)$, satisfies $v=\Gamma u$.\ To that end, let $u\in C_{\rm b}\big([0,T]\times \R\times \mathcal Q\big)$ and $$v\in C^{1,2}\big([0,T)\times \R\times \mathcal Q\big)\cap C_{\rm b}\big([0,T]\times \R\times \mathcal Q\big)$$ such that \eqref{eq:lin.classical} is satisfied with $h=\mathcal Hu$.\ Since $u\in C_{\rm b}\big([0,T]\times \R\times \mathcal Q\big)$, it follows that $\mathcal H u\in C_{\rm b}\big([0,T]\times \R\times \mathcal Q\big)$.\ Fix $(t,x,q)\in[0,T)\times\mathbb{R}\times\mathcal{Q}$ and, as in the proof of Theorem \ref{thm:verification}, set
\[
R_n:=n+|x|,
\qquad
\theta_n:=T-\frac{T-t}{n+1},
\]
and
\[
\tau_n:=\inf\big\{s\in[t,T]:|L_s^{t,x}|\geq R_n\big\}\wedge\theta_n .
\]
By continuity of $L^{t,x}$ on $[t,T]$, we have \(\tau_n\uparrow T\) almost surely.\ Moreover,
\(\partial_xv\) is bounded on
\([t,\theta_n]\times[-R_n,R_n]\times\mathcal{Q}\), and hence
\[
\E\bigg[
  \int_t^{\tau_n}
  \partial_xv(s,L_s^{t,x},q)\sigma(s,L_s^{t,x})\,dW_s
\bigg]=0.
\]
Applying It\^o's formula on \([t,\tau_n]\) and using
\[
\partial_tv(s,y,q)
 +\bigl(\mathcal{L}_sv(s,\,\cdot\,,q)\bigr)(y)
 =h(s,y,q)\qquad \text{for all }y\in \R,
\]
we obtain
\[
\mathbb{E}\bigl[v(\tau_n,L_{\tau_n}^{t,x},q)\bigr]
=
v(t,x,q)
+
\mathbb{E}\bigg[
  \int_t^{\tau_n}h(s,L_s^{t,x},q)\,ds
\bigg].
\]
Since $v$ and $h$ are bounded, dominated convergence and the terminal
condition $v(T,\, \cdot\,,\, \cdot\,)=\varphi$ yield
\[
v(t,x,q)
=
\mathbb{E}\bigl[\varphi(L_T^{t,x},q)\bigr]
-
\int_t^T
\mathbb{E}\bigl[(\mathcal Hu)(s,L_s^{t,x},q)\bigr]\,ds
=
(\Gamma u)(t,x,q).
\]
In particular, every classical solution $v$ to \eqref{eq:cp} satisfies $v=\Gamma v$, i.e., $v$ is a mild solution.\ By the uniqueness of the mild solution, uniqueness of a classical solution to \eqref{eq:cp} follows.

Now, let $v$ be the unique mild solution to the transformed version of \eqref{eq:HJB_PDE_reduced_pred} with terminal condition \eqref{eq:terminal_pred_reduced}.\ Since $[v]_{\alpha}<\infty$ and therefore $[\mathcal Hv]_\alpha<\infty$ by Lemma \ref{lem:operator.calH}, for $t_1,t_2\in [0,T]$ with $t_1\leq t_2$, it follows that
\begin{align*}
|v(t_1,x,q)-v(t_2,x,q)|&\leq \E\big[\big|\varphi\big(L_T^{t_1,x},q\big) - \varphi\big(L_T^{t_2,x},q\big)\big|\big]\\
&\quad+\int_{t_2}^T\E\big[\big|(\mathcal H v)\big(s,L_{s}^{t_1,x},q\big)-\big(\mathcal H v\big)(s,L_{s}^{t_2,x},q)\big|\big] \, ds\\
&\quad +\int_{t_1}^{t_2} \E\big[\big|(\mathcal H v)\big(s,L_{s}^{t_1,x},q\big)\big|\big]\, ds\\
&\leq C_L^\alpha[\varphi]_\alpha|t_1-t_2|^{\alpha/2}\\
&\quad +[\mathcal Hv]_{\alpha} \int_{t_2}^T\E\big[\big|L_{s}^{t_1,x}-L_{s}^{t_2,x}\big|\big]^\alpha\, ds\\
&\quad +\|\mathcal Hv\|_\infty |t_2-t_1|\\
&\leq C_L^\alpha \big([\varphi]_\alpha+T[\mathcal Hv]_{\alpha}\big) |t_2-t_1|^{\alpha/2}+\|\mathcal Hv\|_\infty |t_2-t_1|,
\end{align*}
where, in the second-to-last step, we used Jensen's inequality.\ Since $v$ is bounded, we thus find that $v\in C^{\alpha/2,\alpha}\big([0,T]\times \R\times \mathcal Q\big)$ and, by Lemma \ref{lem:operator.calH}, it follows that $\mathcal Hv\in C^{\alpha/2,\alpha}\big([0,T]\times \R\times \mathcal Q\big)$.

Since $\varphi\in C^{2+\alpha}(\R\times \mathcal Q)$ and $\mathcal Hv\in C^{\alpha/2,\alpha}\big([0,T]\times \R\times \mathcal Q\big)$, by \citep[Theorem 9.2.3]{Krylov1996}, there exists a classical solution
\[
w \in C^{1+\alpha/2,2+\alpha}\big([0,T)\times \R\times \mathcal Q\big)\cap C^{\alpha/2,\alpha}\big([0,T]\times \R\times \mathcal Q\big)
\]
to \eqref{eq:lin.classical} with $h=\mathcal H v$.\ Since $v$ is a mild solution, by the first part of the proof, it follows that $w=v$ is the unique classical solution to \eqref{eq:cp}.

Applying the inverse of $f$, we obtain that  $V(t,p,q):=v\big(t,f^{-1}(p),q\big)$ for $t\in [0,T]$, $p\in (0,1)$, and $q\in \mathcal Q$ yields a classical solution to \eqref{eq:HJB_PDE_reduced_pred} with terminal condition \eqref{eq:terminal_pred_reduced}.\ Since every classical solution to \eqref{eq:HJB_PDE_reduced_pred} with terminal condition \eqref{eq:terminal_pred_reduced} can be  transformed into a classical solution to \eqref{eq:cp}, uniqueness follows.\ The proof is complete. 
\end{proof}


\begin{thebibliography}{23}
\providecommand{\natexlab}[1]{#1}
\providecommand{\url}[1]{\texttt{#1}}
\expandafter\ifx\csname urlstyle\endcsname\relax
  \providecommand{\doi}[1]{doi: #1}\else
  \providecommand{\doi}{doi: \begingroup \urlstyle{rm}\Url}\fi

\bibitem[Abernethy et~al.(2013)Abernethy, Chen, and Vaughan]{AbernethyEtAl2013}
J.~Abernethy, Y.~Chen, and J.~W. Vaughan.
\newblock {E}fficient {M}arket {M}aking via {C}onvex {O}ptimization, and a
  {C}onnection to {O}nline {L}earning.
\newblock \emph{ACM Transactions on Economics and Computation}, 1\penalty0
  (2):\penalty0 Article 12, 2013.

\bibitem[Avellaneda and Stoikov(2008)]{AvellanedaStoikov2008}
M.~Avellaneda and S.~Stoikov.
\newblock {H}igh-{F}requency {T}rading in a {L}imit {O}rder {B}ook.
\newblock \emph{Quantitative Finance}, 8\penalty0 (3):\penalty0 217--224, 2008.

\bibitem[Baldacci et~al.(2021)Baldacci, Bergault, and
  Guéant]{BaldacciBergaultGueant2021}
B.~Baldacci, P.~Bergault, and O.~Guéant.
\newblock Algorithmic market making for options.
\newblock \emph{Quantitative Finance}, 21\penalty0 (1):\penalty0 85--97, 2021.

\bibitem[Barzykin et~al.(2023)Barzykin, Bergault, and
  Guéant]{BarzykinBergaultGueant2023}
A.~Barzykin, P.~Bergault, and O.~Guéant.
\newblock Algorithmic market making in dealer markets with hedging and market
  impact.
\newblock \emph{Mathematical Finance}, 33\penalty0 (1):\penalty0 41--79, 2023.

\bibitem[Berg et~al.(2008)Berg, Nelson, and Rietz]{BergNelsonRietz2008}
J.~E. Berg, F.~D. Nelson, and T.~A. Rietz.
\newblock Prediction market accuracy in the long run.
\newblock \emph{International Journal of Forecasting}, 24\penalty0
  (2):\penalty0 285--300, 2008.

\bibitem[Cartea and Jaimungal(2015)]{CarteaJaimungal2015}
A.~Cartea and S.~Jaimungal.
\newblock Risk {M}etrics and {F}ine {T}uning of {H}igh-{F}requency {T}rading
  {S}trategies.
\newblock \emph{Mathematical Finance}, 25\penalty0 (3):\penalty0 576--611,
  2015.

\bibitem[Cartea et~al.(2014)Cartea, Jaimungal, and
  Ricci]{CarteaJaimungalRicci2014}
A.~Cartea, S.~Jaimungal, and J.~Ricci.
\newblock {B}uy {L}ow, {S}ell {H}igh: A {H}igh {F}requency {T}rading
  {P}erspective.
\newblock \emph{SIAM Journal on Financial Mathematics}, 5\penalty0
  (1):\penalty0 415--444, 2014.

\bibitem[Cartea et~al.(2017)Cartea, Donnelly, and
  Jaimungal]{CarteaDonnellyJaimungal2017}
A.~Cartea, R.~Donnelly, and S.~Jaimungal.
\newblock Algorithmic {T}rading with {M}odel {U}ncertainty.
\newblock \emph{SIAM Journal on Financial Mathematics}, 8\penalty0
  (1):\penalty0 635--671, 2017.

\bibitem[Chen and Pennock(2007)]{ChenPennock2007}
Y.~Chen and D.~M. Pennock.
\newblock A {U}tility {F}ramework for {B}ounded-{L}oss {M}arket {M}akers.
\newblock In \emph{Proceedings of the Twenty-Third Conference on Uncertainty in
  Artificial Intelligence}, UAI'07, page 49–56, Arlington, Virginia, USA,
  2007. AUAI Press.

\bibitem[Cowgill and Zitzewitz(2015)]{CowgillZitzewitz2015}
B.~Cowgill and E.~Zitzewitz.
\newblock {C}orporate {P}rediction {M}arkets: {E}vidence from {G}oogle, {F}ord,
  and {F}irm {X}.
\newblock \emph{The Review of Economic Studies}, 82\penalty0 (4):\penalty0
  1309--1341, 2015.

\bibitem[Dalen(2026)]{Dalen2026}
S.~Dalen.
\newblock {T}oward {B}lack {S}choles for {P}rediction {M}arkets: {A} {U}nified
  {K}ernel and {M}arket {M}aker's {H}andbook.
\newblock \emph{Preprint}, arXiv:2510.15205, 2026.

\bibitem[El~Aoud and Abergel(2015)]{AbergelElAoud2015}
S.~El~Aoud and F.~Abergel.
\newblock {A} {S}tochastic {C}ontrol {A}pproach to {O}ption {M}arket {M}aking.
\newblock \emph{Market Microstructure and Liquidity}, 1\penalty0 (1):\penalty0
  1550006, 2015.

\bibitem[Gu{\'e}ant et~al.(2013)Gu{\'e}ant, Lehalle, and
  Fernandez-Tapia]{GueantLehalleFernandezTapia2013}
O.~Gu{\'e}ant, C.-A. Lehalle, and J.~Fernandez-Tapia.
\newblock Dealing with the {I}nventory {R}isk: A {S}olution to the {M}arket
  {M}aking {P}roblem.
\newblock \emph{Mathematics and Financial Economics}, 7:\penalty0 477--507,
  2013.

\bibitem[Guilbaud and Pham(2013)]{GuilbaudPham2013}
F.~Guilbaud and H.~Pham.
\newblock {O}ptimal {H}igh-{F}requency {T}rading with {L}imit and {M}arket
  {O}rders.
\newblock \emph{Quantitative Finance}, 13\penalty0 (1):\penalty0 79--94, 2013.

\bibitem[Guéant(2017)]{Gueant2017}
O.~Guéant.
\newblock {O}ptimal {M}arket {M}aking.
\newblock \emph{Applied Mathematical Finance}, 24\penalty0 (2):\penalty0
  112--154, 2017.

\bibitem[Hanson(2003)]{Hanson2003}
R.~Hanson.
\newblock {C}ombinatorial {I}nformation {M}arket {D}esign.
\newblock \emph{Information Systems Frontiers}, 5\penalty0 (1):\penalty0
  107--119, 2003.

\bibitem[Ho and Stoll(1981)]{HoStoll1981}
T.~Ho and H.~R. Stoll.
\newblock {O}ptimal {D}ealer {P}ricing under {T}ransactions and {R}eturn
  {U}ncertainty.
\newblock \emph{Journal of Financial Economics}, 9\penalty0 (1):\penalty0
  47--73, 1981.

\bibitem[Jusselin(2021)]{Jusselin2021}
P.~Jusselin.
\newblock {O}ptimal {M}arket {M}aking with {P}ersistent {O}rder {F}low.
\newblock \emph{SIAM Journal on Financial Mathematics}, 12\penalty0
  (3):\penalty0 1150--1200, 2021.

\bibitem[Krylov(1996)]{Krylov1996}
N.~V. Krylov.
\newblock \emph{Lectures on elliptic and parabolic equations in {H{\"o}lder}
  spaces}, volume~12 of \emph{Grad. Stud. Math.}
\newblock Providence, RI: AMS, American Mathematical Society, 1996.

\bibitem[Ng et~al.(2026)Ng, Peng, Tao, and Zhou]{NgEtAl2026}
H.~Ng, L.~Peng, Y.~Tao, and D.~Zhou.
\newblock {P}rice {D}iscovery and {T}rading in {M}odern {P}rediction {M}arkets.
\newblock \emph{Preprint}, SSRN 5331995, 2026.

\bibitem[Nyström et~al.(2014)Nyström, Ould~Aly, and
  Zhang]{NystromOuldAlyZhang2014}
K.~Nyström, S.~M. Ould~Aly, and C.~Zhang.
\newblock {M}arket {M}aking and {P}ortfolio {L}iquidation under {U}ncertainty.
\newblock \emph{International Journal of Theoretical and Applied Finance},
  17\penalty0 (5):\penalty0 1450034, 2014.

\bibitem[{\O}ksendal(2003)]{oksendal}
B.~{\O}ksendal.
\newblock \emph{Stochastic differential equations. {An} introduction with
  applications.}
\newblock Universitext. Berlin: Springer, 6th edition, 2003.

\bibitem[Polgreen et~al.(2007)Polgreen, Nelson, Neumann, and
  Weinstein]{PolgreenEtAl2007}
P.~M. Polgreen, F.~D. Nelson, G.~R. Neumann, and R.~A. Weinstein.
\newblock {U}se of {P}rediction {M}arkets to {F}orecast {I}nfectious {D}isease
  {A}ctivity.
\newblock \emph{Clinical Infectious Diseases}, 44\penalty0 (2):\penalty0
  272--279, 2007.

\end{thebibliography}

\end{document}